\documentclass[journal, onecolumn]{IEEEtran}

\usepackage{bbm,dsfont,mathrsfs,fixmath,amsthm}
\usepackage{amsmath,amssymb,eucal,graphicx}
\usepackage{stmaryrd,cite}
\usepackage{amsmath,amstext,amsfonts,amssymb}
\usepackage{epsfig}
\usepackage{exscale}
\usepackage{enumerate}
\usepackage{mathtools}
\usepackage{multirow}

\usepackage{algorithm}
\usepackage{algpseudocode }
\usepackage[geometry]{ifsym}
\usepackage{ifsym}
\usepackage{amssymb}
\usepackage{amscd}
\usepackage{dsfont}
\usepackage{amsmath}
\usepackage{epsfig}
\usepackage{graphics}
\usepackage{psfrag}
\usepackage{rotating}
\usepackage{amsmath}
\usepackage{amsfonts}
\usepackage{url}
\usepackage{color}
\usepackage[caption=false,font=footnotesize]{subfig}
\usepackage{epstopdf}
\usepackage{amsthm}
\usepackage{pgfplots}
\usepackage{tikz}

\usepackage{mathtools}
\usepackage{makecell}
\usetikzlibrary{arrows,shapes,decorations,backgrounds}

\newtheorem{theorem}{Theorem}
\newtheorem{lemma}[theorem]{Lemma}
\newtheorem{example}{Example}
\newtheorem{definition}{Definition}

\newtheorem{corollary}[theorem]{Corollary}

\newcommand{\cond}{\,\vert\,}

\newfont{\bbb}{msbm10 scaled 500}

\newfont{\bb}{msbm10 scaled 1100}
\newcommand{\CC}{\mbox{\bb C}}
\newcommand{\RR}{\mbox{\bb R}}

\newcommand{\EE}{\mbox{\bb E}}

\newcommand{\av}{{\bf a}}

\newcommand{\dv}{{\bf d}}
\newcommand{\ev}{{\bf e}}

\newcommand{\hv}{{\bf h}}

\newcommand{\pv}{{\bf p}}

\newcommand{\rv}{{\bf r}}

\newcommand{\xv}{{\bf x}}
\newcommand{\yv}{{\bf y}}

\newcommand{\onev}{{\bf 1}}

\newcommand{\Am}{{\bf A}}

\newcommand{\Id}{{\bf I}}

\newcommand{\Lm}{{\bf L}}

\newcommand{\Qm}{{\bf Q}}
\newcommand{\Rm}{{\bf R}}
\newcommand{\Sm}{{\bf S}}

\newcommand{\Um}{{\bf U}}

\newcommand{\Zm}{{\bf Z}}

\newcommand{\Hc}{{\cal H}}
\newcommand{\Ic}{{\cal I}}
\newcommand{\Jc}{{\cal J}}
\newcommand{\Kc}{{\cal K}}

\newcommand{\Nc}{{\cal N}}

\newcommand{\Qc}{{\cal Q}}

\newcommand{\gammav}{\hbox{\boldmath$\gamma$}}
\newcommand{\deltav}{\hbox{\boldmath$\delta$}}

\newcommand{\lambdav}{\hbox{\boldmath$\lambda$}}

\newcommand{\nuv}{\hbox{\boldmath$\nu$}}
\newcommand{\muv}{\hbox{\boldmath$\mu$}}

\newcommand{\sigmav}{\hbox{\boldmath$\sigma$}}

\renewcommand{\arg}{{\hbox{arg}}}

\DeclareFontFamily{U}{cmfi}{}
\DeclareFontShape{U}{cmfi}{m}{n}{ <-> cmfi10 }{}
\DeclareSymbolFont{CMFI}{U}{cmfi}{m}{n}

\def\argmax{\mathop{\rm argmax}}

\renewcommand{\Am}{\pmb{A}}

\renewcommand{\Lm}{\pmb{L}}

\renewcommand{\Qm}{\pmb{Q}}
\renewcommand{\Rm}{\pmb{R}}
\renewcommand{\Sm}{\pmb{S}}

\renewcommand{\Um}{\pmb{U}}

\renewcommand{\Zm}{\pmb{Z}}

\renewcommand{\av}{\pmb{a}}

\renewcommand{\dv}{\pmb{d}}
\renewcommand{\ev}{\pmb{e}}

\renewcommand{\hv}{\pmb{h}}

\renewcommand{\pv}{\pmb{p}}

\renewcommand{\rv}{\pmb{r}}

\renewcommand{\xv}{\pmb{x}}
\renewcommand{\yv}{\pmb{y}}

\newcommand{\gp}[1]{{#1}}

\newcommand{\ind}[1]{\mathds{1}_{\left\lbrace #1 \right\rbrace}}

\newcommand{\bigO}[1]{\ensuremath{\mathop{}\mathopen{}O\mathopen{}\left(#1\right)}}

\IEEEoverridecommandlockouts
\usepackage{ragged2e}

\begin{document}
\title{Adaptive Coded Caching for Fair Delivery over Fading Channels}
\author{Apostolos~Destounis,    ~\IEEEmembership{Member,~IEEE,}
	Asma~Ghorbel, 
	Georgios~S.~Paschos,
	~and~Mari~Kobayashi,~\IEEEmembership{Senior Member,~IEEE,}  

	\thanks{A. Destounis is  with the Mathematical and Algorithmic Sciences Lab, France Research Center - Huawei Technologies Co. Ltd., 20 quai de Point du Jour, 92100 Boulogne-Bilancourt, France. Email: apostolos.destounis@hauwei.com}
	\thanks{A. Ghorbel was with the Laboratoire des Signaux et Syst\`emes (L2S), CentraleSup\'elec, Universit\'e Paris-Saclay, 3, Rue Joliot-Curie, 91192 Gif sur Yvette, France. Email: asmaghorbel01@gmail.com}
	\thanks{Georgios S. Paschos was   with the Mathematical and Algorithmic Sciences Lab, France Research Center - Huawei Technologies Co. Ltd., 20 quai de Point du Jour, 92100 Boulogne-Bilancourt, France. Currently, he is with Amazon, Luxembourg. E-mail: gpaschos@gmail.com.}  
	\thanks{M. Kobayashi is with the Laboratoire des Signaux et Syst\`emes (L2S), CentraleSup\'elec, Universit\'e Paris-Saclay, 3, Rue Joliot-Curie, 91192 Gif sur Yvette, France. Email: mari.kobayashi@centralesupelec.fr}
	\thanks{Part of the work in this paper has been presented at the 15th International Symposium on Modeling and Optimization in Mobile, Ad Hoc, and Wireless Networks (WiOpt), Telecom ParisTech, Paris, France, 15th - 19th May, 2017.}
}
\maketitle

\begin{abstract}
The performance of existing \emph{coded caching} schemes is sensitive to the worst channel quality, a problem which is exacerbated when communicating over fading channels.
In this paper, we address this limitation in the following manner:  \emph{in short-term}, we allow transmissions to subsets of users with good channel quality, avoiding users with fades, while \emph{in long-term} we ensure fairness among users. 
Our online scheme combines (i) the classical decentralized coded caching scheme \cite{maddah2013decentralized} with (ii) joint scheduling and power control for the fading broadcast channel, as well as (iii) congestion control for ensuring the optimal long-term average performance. 
We prove that our online delivery scheme maximizes the alpha-fair utility among all schemes restricted to decentralized placement. 
By tuning the value of alpha, the proposed scheme can achieve different operating points on the average delivery rate region and tune performance according to an operator's choice. 

We demonstrate via simulations that our scheme outperforms two baseline schemes: (a) standard coded caching with multicast transmission, limited by the worst channel user yet exploiting the global caching gain;  
(b) opportunistic scheduling with unicast transmissions {exploiting the fading diversity but limited to local caching gain}. 
\end{abstract}

\section{Introduction}\label{sec:intro}

A key challenge {for the}  future wireless networks is the increasing video traffic demand, which reached 70\% of total mobile IP traffic in 2015 \cite{cisco15}.
Classical downlink systems cannot meet this demand since they have limited resource blocks, and therefore as the number  $K$ of simultaneous video transfers  increases,  the  per-video throughput  vanishes as $1/K$.
Recently it was shown that \gp{scalable} per-video throughput can be achieved if the communications are synergistically designed with caching at the receivers. 
Indeed, the recent  breakthrough of \emph{coded caching} \cite{maddah2013fundamental} has inspired a  rethinking of  wireless downlink. 
Different video sub-files are cached at the receivers, \gp{and video requests are served by coded multicasts.} 
By careful selection of sub-file caching \gp{and exploitation of the  wireless broadcast channel}, the transmitted signal is simultaneously
useful for decoding at users who requested different video files. 
This scheme has been theoretically proven to scale well, and therefore has the potential to resolve the challenge of downlink bottleneck for future networks. Nevertheless, several limitations hinder its applicability in practical systems \cite{misconceptions}. In this work, we take a closer look to the limitations that arise from the fact that \emph{coded caching was originally designed for a symmetric error-free shared link.} 

\gp{If instead we consider a realistic wireless channel, 
} we observe that coded caching faces a \emph{short-term} limitation.
Namely, its performance is limited by the user in the worst channel condition because the wireless multicast capacity is determined by the worst user \cite[Chapter 7.2]{el2011network}. 
This is in stark contrast with standard downlink techniques such as \emph{opportunistic scheduling} \cite{stolyar,Li05, knopp}, which serve the user with the best instantaneous channel quality. 
Thus, a first challenge is to modify coded caching for exploitation of fading peaks, similar to the opportunistic scheduling. 

In addition to the fast fading consideration, there is also a \emph{long-term} limitation due to a network topology. Namely, the ill-positioned users, e.g. users at the cell edge, may experience long spells of  poor channel quality throughout the video delivery. 
The classical coded caching scheme is designed to provide video files at equal data rates to all users, 
which leads to ill-positioned users consuming most of the air time and hence driving the
overall system performance to low efficiency.
In the literature of wireless scheduling without caches at receivers, this problem has been resolved by the use of fairness among user throughputs \cite{Li05}. 
By allowing poorly located users to receive less throughput than others, precious air time is saved and the overall system performance is greatly increased.
Since the sum throughput rate and equalitarian fairness are typically the two extreme objectives, past works have proposed the use of alpha-fairness \cite{mowalrand} which allows  to select the coefficient $\alpha$ and drive the system to any desirable tradeoff point in between of the two extremes. 
Previously, the alpha-fair objectives have been studied in the context of (i) multiple user activations \cite{stolyar}, (ii) multiple antennas \cite{caire_fairnessMIMO} and (iii) broadcast channels \cite{caire_fairnessBC}. However, in the presence of caches at user terminals, the fairness problem is further complicated by the interplay between user scheduling and designing codewords for multiple users. In particular, we wish to shed light into the following questions: \emph{Which user requests shall we combine together to perform coded caching? How shall we schedule a set of users to achieve our fairness objective while adapting to time-varying channel quality?}

To address these questions, we study the content delivery over a realistic block-fading broadcast channel, where the channel quality varies across users and time. 
Although the decisions of user scheduling and codeword design are inherently coupled, we design a scheme which decouples these two problems, while maintaining optimality through a specifically designed queueing structure. On the transmission side, we select the multicast user set dynamically depending on the instantaneous channel quality and user urgency captured by queue lengths. On the coding side, we adapt the codeword construction of \cite{maddah2013decentralized} to the set of users chosen by the appropriate routing which depends also on the past transmission side decisions.
Combining with an appropriate congestion controller, we show that this approach yields our alpha-fair objective {under the restriction that the entire request sequences for each user are different. This restriction is done for technical reasons, however, and it may not be necessary if the caching scheme proposed in \cite{Yu2018_exactTradeoff} is used}.

More specifically, our approaches and contributions are summarized below:
\begin{itemize}
	\item[1)] We design a novel queueing structure which decouples the channel scheduling from the codeword construction. Although it is clear that the codeword construction needs to be adaptive to channel variation, our scheme ensures this through our \emph{backpressure} that connects the user queues and the codeword queues. Hence, we are able to show that this decomposition is without loss of optimality (see Theorem \ref{th:optimality_infinite}).
	\item[2)] We then provide an online policy consisting of (i) admission control of new files into the system; (ii) combination of files to perform coded caching; (iii) scheduling and power control of codeword transmissions to subset of users on the wireless channel. We prove that 
	the long-term video delivery rate vector achieved by our scheme is a near optimal solution to the alpha-fair optimization problem under the restriction to policies that are based on the decentralized coded caching scheme \cite{maddah2013decentralized}. {That is, the optimality is within the set of policies that are restricted to use caching from \cite{maddah2013decentralized} but are free to schedule transmissions, admit video files, and tune the power in any unrestricted way.}
	\item[3)] Through numerical examples, we demonstrate the superiority of our approach versus (a) standard coded caching with multicast transmission limited by the worst channel condition yet exploiting the global caching gain, 
	(b) opportunistic scheduling with unicast transmissions exploiting only the local caching gain. 
	This shows that our scheme not only is the best among online decentralized coded caching schemes, but moreover manages to exploit opportunistically the time-varying fading channels.
\end{itemize}

\subsection{Related work}

Since coded caching was first introduced in \cite{maddah2013fundamental} and its potential was recognized by the community, 
substantial efforts have been devoted to quantify the gain in realistic scenarios, including decentralized
placement \cite{maddah2013decentralized}, non-uniform popularities \cite{ji2017order,niesen2017coded}, and more general network topologies (e.g. \cite{ji2013fundamental,karamchandani2016hierarchical,shariatpanahi2016multi}). 
A number of recent works have studied coded caching by replacing the original perfect shared link with wireless channels \cite{huang2015performance,zhang2016wireless,zhang2017fundamental,bidokhti2016noisy,NgoAllerton2016,shariatpanahi2017multi}. 
In particular, the harmful effect of coded caching over wireless multicast channels has been highlighted recently \cite{combes2018utility,huang2015performance,zhang2016wireless,NgoAllerton2016}, while similar conclusions and some directions are given in \cite{combes2018utility,huang2015performance,zhang2016wireless,bidokhti2016noisy}. Although \cite{combes2018utility} consider the same
channel model and address a similar question as in the current work, they differ in their objectives and approaches. \cite{combes2018utility} highlights the scheduling part and provides rigorous analysis on the long-term average per-user rate in the regime of large number of users. In the current work, a new queueing structure is proposed to deal jointly with admission control,
routing, as well as scheduling for a finite number of users.

Furthermore, most of existing works have focused on {\it offline} caching where both cache placement and delivery phases are performed once without capturing the random and asynchronous nature of video traffic. 
The works \cite{pedarsani2016online, niesen2015coded} addressed partly the online aspect by studying cache eviction strategies, the delivery delay, respectively. In this work, we will explore a different online aspect. Namely, we assume that 
the file requests from users arrive dynamically and the file delivery is performed continuously over time-varying fading broadcast channels. {For the case of dynamic file requests, authors in \cite{neely13_indexCoding} derived a MaxWeight type of rule over a pre-specified set of coding actions. The focus of that work is, however, on delivering all requested file demands and neither fairness nor time varying channels are considered.}

Finally, online transmission scheduling over wireless channels has been extensively studied in the context of opportunistic scheduling \cite{stolyar} and network utility maximization \cite{neely10}. Prior works emphasize two fundamental aspects: (a) the balancing of user rates according to fairness and efficiency considerations, and (b) the opportunistic exploitation of the time-varying fading channels. 
{There have been some works that study scheduling policies over a queued-fading downlink channel}; \cite{Seong06} gives a maxweight-type of policy and \cite{Eryilmaz01} provides a throughput optimal policy based on a fluid limit analysis.
Our work is the first to our knowledge that studies coded caching in this setting. 
The new element in our study is the joint consideration of user scheduling with codeword construction for the coded caching delivery phase. 

\section{System Model and Problem Formulation}\label{sec:Model}

\subsection{Overview of the System Model}
We consider a time-slotted content delivery network where a server (or a base station) wishes to convey requested files to $K$ user terminals over a wireless channel. Each device has a cache where it can store parts of the files. In this following subsection, we provide an overview of our system model while the details of queuing structure, file combining and codeword generation, and transmissions over the wireless channel are presented in Sections \ref{ssec:queueingStructure}, \ref{ssec:MANscheme}, and \ref{ssec:WirelessModel}, respectively. 

At the beginning of slot $t$, user $k$ makes requests for $A_k(t)$ files, chosen uniformly at random among the $N$ files of the catalog. The files are named as $W_1, \dots, W_N$, and each file is $F$ bits long. The request sequence is i.i.d. in time. However, we begin the paper with the ``infinite reservoir" assumption, which means the number of requested files at each slot is so large that it is not possible to deliver it to the user by the end of the time slot (technically that the mean request vector is outside the feasibility region, please refer to Section III.A for more details). Alternatively, it is equivalent to assuming that there are always enough files that can be sent to a user. This assumption is essentially the same as the standard ``infinite backlog" model for network utility maximization problems (see e.g. \cite[Section III]{neely10}, \cite{Georgiadis06}). We will lift this assumption in Section \ref{sec:dynamic}. 

The wireless channel is modeled by a standard block-fading broadcast channel, such that the channel state remains constant over a slot and changes from one slot to another in an i.i.d. manner. Let $\phi_{\hv}$ the probability that the vector of the channel states of all users is $\hv$. Each slot is assumed to allow for $T_{\rm slot}$ channel uses. The channel output of user $k$ in any channel use of slot $t$ is given by  
\begin{align}\label{eq:BFC}
\yv_k(t)=\sqrt{h_k(t)} \xv_k(t)+\nuv_k(t),
\end{align}
where the channel input $\xv_k\in\CC^{T_{\rm slot}}$ is subject to the power constraint $\EE[\Vert \xv \Vert^2] \leq PT_{\rm slot}$; $\nuv_{k}(t)\sim\Nc_{\CC}(0, \Id_{T_{\rm slot}})$ are additive white Gaussian noises with covariance matrix identity of size $T_{\rm slot}$, assumed independent of each other; $\{h_k(t)\in\CC\}$ are channel fading coefficients independently distributed across time. At each slot $t$, the channel state $\hv(t)=(h_1(t), \dots, h_K(t))$ is perfectly known to the base station while each user knows its own channel realization. 

Each user $k$ is equipped with cache memory $Z_k$ of $MF$ bits, where  $M\in \{0,1,\dots,N\}$. 
We restrict ourselves to decentralized cache placement \cite{maddah2013decentralized}. More precisely, 
each user $k$ independently caches a subset of $\frac{MF}{N}$ bits of file $i$, chosen uniformly at random for $i=1,\dots, N$, under its memory constraint of 
$MF$ bits. For later use, we let $m=\frac{M}{N}$ denote the normalized memory size. 
By letting $W_{i|\Jc}$ denote the sub-file of $W_i$ stored exclusively in the cache memories of the user set $\Jc \subseteq \{1,2,...,K\}$, and $W_{i|\emptyset}$ denoting the sub-file of $W_i$ stored in none of the users (throughout the paper we will also use an alternative notation, denoting files with capital letters and $A_{\Jc}$ the sub-file of file $A$ that is stored exclusively in all users of set $\Jc$. ), the cache memory $Z_k$ of user $k$ after decentralized placement is given by
\begin{align} \label{eq:Zk}
Z_k =\{ W_{i \cond \Jc}:  \;\;  \Jc \subseteq\{1, \dots, K\}, k\in\Jc , \forall i =1,\dots, N \}.
\end{align} 
As in most previous works, we make the assumption of large file size ($F\rightarrow \infty$), therefore the size of each sub-file (measured in bits) can be well approximated by the law of large numbers as   
\begin{align}
|W_{i \cond \Jc}|= m^{|\Jc|}\left(1-m\right)^{K-|\Jc|} F \label{eq:LLN}
,\end{align}
where $|.|$ denotes here the size of the subfile. Note that we will assume that this placement is performed once, before the beginning of the system's operation and stays constant. \footnote{In practice caches would be populated at off-peak hours and stay constant during a large period of time, e.g. populated at night and stay constant during the day} 

Placing parts of a file to each user cache gives the opportunity to exploit the broadcast nature of the wireless medium in order to send useful information to multiple of users. For example, assume users $1,2$ have made requests for files $W_1, W_2$ and $W_3, W_4$, respectively, and these files are at queues waiting to be transmitted. One way to deliver those files would be to make single-user transmissions, and transmit, at each slot, to each user the missing parts of files they requested. Alternatively, we could combine the files for both users by combining, for example, files $1, 3$ and $2,4$. In this case, we can deliver the files by performing single user transmissions for the codewords $\{W_{i|\emptyset}\}_{i=1,2,3,4}$ (the parts of the files no user has stored) and multicast transmissions to both user for the codewords $W_{1 | \{2\}} \oplus W_{3 | \{1\}}$ and $W_{2, | \{2\}} \oplus W_{4 | \{1\}}$, where $\oplus$ denotes the bit-wise XOR operation. The latter scheme needs to transmit fewer information bits over the channel and may be preferable if/when both users have simultaneously good channel conditions. As file requests arrive dynamically and the channel states change over time, we need to continuously decide how to send the backlogged files: Sometimes it may be worth combining files that are requested among a set of users into codewords, that is bit sequences that have to be delivered to appropriate subsets of that set of users. In addition, we should decide when to send, over the wireless channel, bits from codewords created from such combinations (and if yes, send codeword bits corresponding to which subset of users) according to the channel realization; it is beneficial to attempt multicast transmissions to groups of users that have simultaneously a good channel realization.  

The subsection that follows casts this problem in a queueing framework, while subsequent subsections detail how files are combined and how transmissions are performed over the wireless channel. 

\subsection{High Level System Operation and Queueing Structure} \label{ssec:queueingStructure}

The system operates as follows at each time slot $t$: (1) A number of files are admitted for each user and admitted files wait in a queue for each different user to be processed. (2) The transmitter chooses subsets of admitted but unprocessed files to combine. In our context, combining files to a user subset $\Jc \subseteq \{1,2,...,K\}$  means that a "codeword" of bits that have to be delivered via a multicast transmission will be generated for each user subset  $\Ic \subseteq \Jc$. These codewords (the respective bits) are stored in queues, with one queue for each user subset, where queue for subset $\Ic$ is served by a multicast transmission towards all users in set $\Ic$. (3) Transmission of codewords is performed over the Gaussian channel in an opportunistic fashion in order to exploit channel realizations that are favorable for transmitting to user subsets. 

This operation is illustrated in Fig. \ref{fig:queueing_system_real} for $K=3$ users.

\begin{figure}[t]
	\centering
	\includegraphics[scale=0.6]{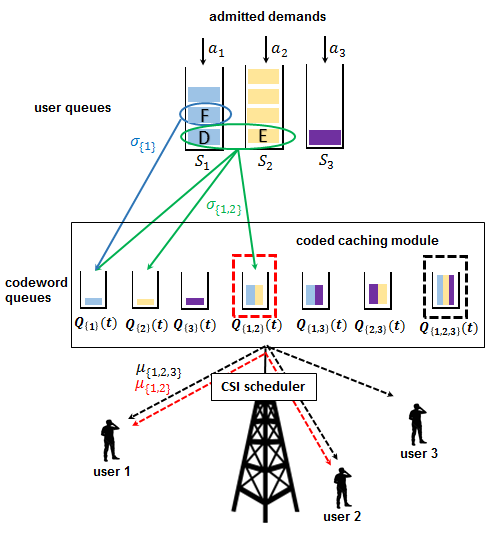}
	\caption{An example of the queueing model for a system with $3$ users. Dashed lines represent wireless transmissions, solid circles files to be combined and solid arrows codewords generated.}
	\label{fig:queueing_system_real}
\end{figure}

At each time slot $t$, the controller \emph{admits} $a_k(t)$ files to be delivered to user $k$, and hence $a_k(t)$ is a control variable. We equip the base station with the following types of queues, as illustrated in Fig. \ref{fig:queueing_system_real}: 

\begin{enumerate}
	\item \textbf{User queues}  to store admitted files, one for each user. The buffer size of queue $k$ is denoted by $S_k(t)$ and expressed in number of files. 
	\item \textbf{Codeword queues} to store codewords to be multicast. There is one codeword queue for each user subset $\Ic \subseteq \{1, \dots, K\}$, thus we have $2^K-1$ such queues. The size of codeword queue $\Ic$ is denoted by $Q_{\Ic}(t)$ and expressed in bits.
\end{enumerate} 

A queueing policy $\pi$ performs the following operations: (i) it decides how many files to admit into  the user queues $S_k(t)$ using $a_k(t)$ variables, (ii) it combines files destined to different users to create multiple codewords. We denote by  $\sigma_\Jc$ the codeword routing variable that defines the number of files combinations among requested files by users in subset $\Jc$. The file combination operates according to the coded caching scheme in \cite{maddah2013decentralized}. We consider $\sigma_\Jc \in \{0,..,\sigma_max\}$ where $\sigma_max$ is he muximum number of files combinations for each user subset.(iii) it decides the encoding function for the wireless transmission. Next, we focus on the queue operations, while the wireless transmissions are explained in \ref{ssec:WirelessModel}

\begin{enumerate}
	\item \textbf{Admission control:}
	At the beginning of each slot, the controller decides how many requests for each user, $a_k(t)$ should be pulled into the system from the infinite reservoir.

	\item \textbf{Codeword Routing:}
	The admitted files for user $k$ are stored in queues $S_k(t)$ for $k=1,\dots, K$. At each slot, files from subsets of these queues are combined into codewords by means of the decentralized coded caching encoding scheme. The files combinations are decided through the routing variable $\sigma_\Jc$ for each user subset $\Jc$. For example: deciding on $\sigma_{1,2}$ to be equal to $1$,  means that we combine one requested file by user $1$ (from user queue $S_1$) with one requested file by user $2$ (from user queue $S_2$) according to the coded caching scheme in \cite{maddah2013fundamental} for $2$ users. This will generate $3$ different subfiles intended to: only user $1$, only user $2$ or simultaneously to users $1$ and $2$. If $\sigma_{1,2}=n>1$, we repeat the previous operation $n$ times for different files.  Depending on the destination, the created subfiles are stored in the corresponding codeword queue (as mentioned in the following operation). \footnote{It is worth noticing that standard coded caching uses $\sigma_{\Jc} = 1$ for $\Jc =\{1,\dots, K\}$ and zero for all the other subsets ($\sigma_{\max}$ is the maximum number of requests that can be combine and is a system parameter). On the other hand, uncoded caching can be represented by $\sigma_{\Jc}=1$ for $\Jc={k}, k\in {1,....,K}$. Our scheme can, therefore, be seen as a combination of both, which explains its better performance.} The size of the user queue $S_k$ evolves as: 
	\begin{align}
	S_k(t+1) = \big[S_k(t) - \underbrace{\sum_{\Jc: k\in \Jc}\sigma_{\Jc}(t)}_{\begin{subarray}{c}\text{number of files} \\ \text{combined into} 
		\\ \text{codewords}\end{subarray}}\big]^+ + \underbrace{a_k(t)}_{\begin{subarray}{c}\text{number of} \\ \text{admitted files}\end{subarray}}  \label{eq:userQueues}
	\end{align}
	If $\sigma_{\Jc}(t)> 0$, the server creates codewords by applying the coded caching framework of \cite{maddah2013decentralized} (see also the next subsection) 
	for this subset of users 
	as a function of the cache contents $\{Z_j: j\in \Jc\}$. 
	
	\item \textbf{Scheduling:}
	The codewords intended to the subset $\Ic$ of users are stored in codeword queue whose size is given by $Q_{\Ic}(t)$ for $\Ic\subseteq \{1,\dots, K\}$. Given the instantaneous channel realization $\hv(t)$ and the queue state $\{Q_{\Ic}(t)\}$, the server performs multicast scheduling and rate allocation. Namely, at slot $t$, it determines the number $\mu_{\Ic}(t)$ of bits per channel use to be transmitted for the users in subset $\Ic$. By letting $b_{\Jc,\Ic}$ denote the number of bits generated for codeword queue $\Ic\subseteq \Jc$ when coded caching
	is performed to the users in $\Jc$, codeword queue $\Ic$ evolves as 
	\begin{align}\label{eq:codewordQueues}
	Q_{\Ic}(t+1) =  \big[Q_{\Ic}(t) -\underbrace{ T_{\rm slot}\mu_{\Ic}(t)}_{\begin{subarray}{c}  \text{number of bits}\\  \text{multicast to $\Ic$} \end{subarray}}\big]^+  + \underbrace{\sum_{\Jc:\Ic\subseteq\Jc}b_{\Jc,\Ic}\sigma_{\Jc}(t)}_{\begin{subarray}{c}\text{number of bits}\\ \text{created  by} \\ \text{combining files}\end{subarray}}
	.\end{align}	
\end{enumerate}  

\noindent A control policy is, therefore, fully specified by giving the rules with which the decisions $\{\av(t), \sigmav(t), \muv(t)\}$ are taken at every slot $t$. We will design a control policy and characterize its performance in Sections \ref{sec:proposed}, \ref{sec:dynamic}, and in the remaining of the current Section we will detail our assumptions for codeword generation and  wireless transmissions.   

\subsection{File Combining and Codeword Generation Scheme}\label{ssec:MANscheme}

In this work, we use the file encoding scheme of Maddah-Ali and Niesen \cite{maddah2013fundamental, maddah2013decentralized} in order to combine files for the selected sets of users. When a set of users $\Jc$ for which combination of unprocessed files is decided to be done, the offline scheme proceeds to generating the corresponding codewords to be multicast. Assuming that the file for user $k\in\Jc$ to be processed this way is $W_k$, the server generates and conveys the following codeword simultaneously useful to each subset  $\Ic \subset \Jc$:  
\begin{align}
V_{\Ic}=\oplus_{k\in \Ic}W_{k|\Ic\setminus\{k\}} 
.\end{align}
We can show that the size of subfile intended to user subset $\Ic$, after combining files requested by users in $\Jc$, is given by 
\begin{equation}\label{eq:b}
b_{\Jc,\Ic}=m^{|\Ic|}(1-m)^{|\Jc|-|\Ic|-1}F, \Ic \subseteq \Jc
.\end{equation}

It is worth noticing that, the {\it coded} delivery with XORs significantly reduces the number of transmissions. 
Compared to uncoded delivery, where the sub-files are sent sequentially and the number of transmissions are equal to $|\Jc|\times|W_{k|\Jc\setminus\{k\}}|$, the coded delivery requires the transmission of $|W_{k|\Jc\setminus\{k\}}|$, yielding a reduction of a factor $|\Jc|$.
In a practical case of $N>K$, it has been proved that decentralized coded caching achieves the total number of transmissions, measured in the number of files, given by \cite{maddah2013decentralized}
\begin{align}\label{eq:mad13}
T_{\rm tot}(K,m) = \frac{1}{m} \left(1-m\right) \left\{1-\left(1-m\right) ^K \right \}. 
\end{align}
On the other hand, in uncoded delivery, the number of transmissions is given by $K(1-m)$ since it exploits only {\it local} caching gain at each user. For a system with $K=30$ users and normalized memory of $m=1/3$, the minimum transmissions required by uncoded delivery is $20$ and that of decentralized coded caching is $2$, yielding a gain of factor $10$.  

In order to further illustrate the placement and codeword generation we consider in this paper,  we provide a three-user example, illustrated in Fig. \ref{fig:ex3}. 
\begin{example}\label{ex:codedCaching}
	Consider that we want to combine the admitted files $A, B, C$ for users 1, 2, 3, respectively. 
	After the placement phase, a given file $A$ will be partitioned into 8 subfiles.
	\begin{figure}
		\vspace{-8pt}
		\begin{center}
			\includegraphics[width=0.4\textwidth,clip=]{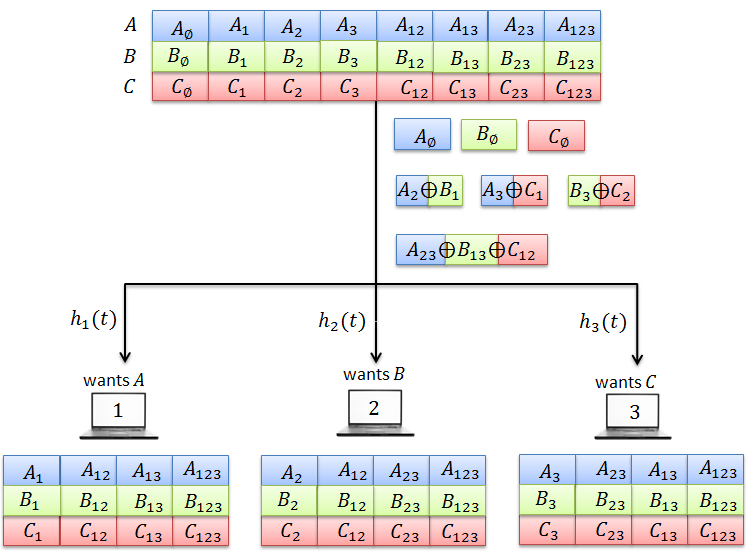}
			\vspace{-2pt}
			\caption{Example of codeword generation for a set of $3$ users.}
			\label{fig:ex3}
		\end{center}
		\vspace{-20pt}
	\end{figure}
	The codewords to be pushed to their respective codeword queues are the following:
	\begin{itemize}
		\item  $A_{ \emptyset}$, $B_{ \emptyset}$ and $C_{ \emptyset}$ to user $1$, $2$ and $3$ respectively.  
		\item  $A_{ 2}\oplus B_{ 1}$ is intended to users $\{1,2\}$. Once received, user $1$ decodes $A_{ 2}$ by combining the received codeword with $B_{ 1}$ given in its cache. Similarly user $2$ decodes $B_{ 1}$. The same approach holds for codeword $B_{ 3}\oplus C_{ 2}$ to users $\{2,3\}$ and codeword $A_{ 3}\oplus C_{ 1}$ to users $\{1,3\}$
		\item  $A_{ 23}\oplus B_{ 13}\oplus C_{ 12}$ is intended users ${1,2,3}$. User $1$ can decode $A_{ 23}$ by combining the received codeword with $\{B_{ 13},C_{ 12}\}$ given in its cache. The same approach is used for user $2$, $3$ to decode $B_{ 13}$, $C_{ 12}$ respectively.  
	\end{itemize}     
\end{example}

\subsection{Model for Transmissions over the Wireless Channel}\label{ssec:WirelessModel}

The channel in \eqref{eq:BFC} for a given realization $\hv$ is a stochastically degraded Broadcast Channel (BC) which achieves the same capacity region as the physically degraded BC \cite[Sec. 5]{el2011network}. The capacity region of the degraded broadcast channel for $K$ private messages and a common message is well-known \cite{el2011network}. However, in our case, due to the combination of files and the generation of codewords for subsets of users as explained in the previous subsection, we have an extended setup where the transmitter wishes to convey  $2^K-1$ mutually independent messages, denoted by $\{M_{\Jc}\}$, where $M_{\Jc}$ denotes the message intended to the users in subset $\Jc\subseteq \{1,\dots, K\}$. We require that each user $k$ must decode all messages $\{M_{\Jc}\}$ for $\Jc\ni k$. By letting $R_{\Jc}$ denote the multicast rate of the message $M_{\Jc}$, we say that the rate-tuple $\Rm\in \RR_+^{2^K-1}$ is achievable if there exists encoding and decoding functions which ensure the reliability and the rate condition as the slot duration $T_{\rm slot}$ is taken arbitrarily large. The capacity region is defined as the supremum of the achievable rate-tuple as shown in \cite{combes2018utility}, where the rate is measured in bit/channel use.   
\begin{theorem}\label{th:capacityDegradedBC}
	The capacity region $\Gamma(\hv)$ of a $K$-user degraded Gaussian broadcast channel with fading gains $h_1 \geq \dots \geq h_K$ and $2^K-1$ independent messages $\{M_{\Jc}\}$ is given by 
	\begin{equation}\label{eq:capRegion}
	\sum_{\Jc \subseteq \{1,\dots, k\}: k\in\Jc} R_{\Jc} \leq \log\frac{1+ h_k \sum_{j=1}^{k} p_j }{1+ h_k\sum_{j=1}^{k-1} p_j } \;\;\; k=1, \dots, K\end{equation}
	for non-negative variables $\{p_k\}$ such that $\sum_{k=1}^K p_k \leq P$. 
\end{theorem}
\begin{proof}
	The proof is quite straightforward and is based on rate-splitting and the private-message
	region of degraded broadcast channel. For completeness, see details in Appendix \ref{appendix:superp}.
\end{proof}

The achievability builds on superposition coding at the transmitter and successive interference cancellation at receivers. For $K=3$, the transmit signal is simply given by 
\[x = x_1 + x_2 + x_3 + x_{12} + x_{13}+ x_{23}+ x_{123},
\]
where $x_{\Jc}$ denotes the signal corresponding to the message $M_{\Jc}$ intended to the subset $\Jc\subseteq \{1,2,3\}$. We suppose that all $\{x_{\Jc}: \Jc\subseteq\{1,\dots,K\}\}$ are mutually independent Gaussian distributed random variables satisfying the power constraint. User 3 (the weakest user) decodes $\tilde{M}_3 =\{M_{3}, M_{13}, M_{23}, M_{123}\}$ by treating all the other messages as noise. User 2 decodes first the messages $\tilde{M}_3$ and then jointly decodes $\tilde{M}_2 =\{M_2, M_{12}\}$. Finally, user 1 (the strongest user) successively decodes $\tilde{M}_3, \tilde{M}_2$ and, finally, $M_1$. 

For the rest of the paper we will assume that the blocklength for transmission $T_{slot}$ is large enough so that transmissions at any rate within the capacity region $\Gamma(\hv)$ using the above achievability scheme  is indeed feasible with no errors.  Here we clarify that although we showed how to achieve the information-theoretic capacity region of a broadcast channel with fading, this does not solve our problem, but only clarifies how to use power control for the different codeword transmissions over the broadcast fading channel.

In order to correctly decode a requested file, a user $k$ must receive the codewords  generated for all subsets containing user $k$ when this file was combined. 

\begin{example}[A simple system operation]
	We conclude this Section with an example of the system operation for sime slots for $K=3$ users. Assume that the system starts ($t=0$) with only files $A, B, C$ waiting in the user queue for users 1,2,3 and are combined as in Example \ref{ex:codedCaching}. Then, at the start of $t=1$ the codewords queues will be filled with the bits of the corresponding subfiles shown in Fig. \ref{fig:ex3}. In that slot, assume that files $D$, $E$ and $F$ are admitted for users 1 and 2, respectively, and that codewords $A_{23}\oplus B_{12} \oplus C_{12}$ and  $A_{2}\oplus B_{1}$ are delivered. Then, at the beginning of slot $t=2$, the content of the user queues will be $D, E, F$ and that of the codeword queues: $A_{\emptyset}, B_{\emptyset}, C_{\emptyset}, \text{empty}, A_{3}\oplus C_{1}, B_{3}\oplus C_{2}, \text{empty}$ (user and codeword queues are ordered  as in Fig. \ref{fig:queueing_system_real} from left to right). At this slot, no admissions are made, no combinations and all single user transmissions are made, emptying queues for ${1}, {2}$ and leaving half of the content of queue ${3}$. At the next slot $t=3$, file $G$ for user 2 is admitted, a transmission to subsets ${3}, {13}, {23}$ are made so the queues are emptied and files for users 1 and 3 are combined. At the end of that slot, users 1,2,3 will be in a position to decode files $A, B, C$ and the queue states at the beginning of $t=4$ will be $\text{empty}, \{E, G\}, \text{empty}$ for user queues and $ D_{\emptyset}, \text{empty}, F_{\emptyset}, \text{empty}, \text{empty}, D_{3}\oplus F_{1}\text{empty}$ for codeword queues.      
\end{example}

\section{The Fair File Delivery Problem}
After having specified the operation of the system in high level, this section will formulate the problem of alpha-fair file delivery. 

Now we are ready to define the feasible rate region as the set of the average number of successfully delivered files for $K$ users. We let $\overline{r}_k$ denote \emph{time average delivery rate} of user $k$, measured in files per slot. We let $\Lambda$ denote the set of all feasible delivery rate vectors. 
\begin{definition}[Feasible rate]
	A rate vector $\overline{\rv}=(\overline{r}_1,\dots, \overline{r}_K)$, measured in file/slot, is said to be {\it feasible} $\overline{\rv}\in\Lambda$ if there exist a file combining and transmission scheme such that  
	\vspace{-0.05in}
	\begin{align}\label{eq:feasiblerate}
	\overline{{r}}_k=\liminf_{t\rightarrow\infty} \frac{\overline{D}_k(t)}{t}.
	\end{align}
	where $\overline{D}_k(t)$ denotes the number of successfully delivered files to user $k$ up to $t$. 
\end{definition}
It is worth noticing that as $t\rightarrow \infty$ the number of decoded files $\hat{D}_k(t)$ shall coincide with the number of successfully delivered files $\overline{D}_k(t)$ under the assumptions discussed previously. 
In contrast to the original framework \cite{maddah2013fundamental,maddah2013decentralized}, which focused on worst-case demands, our rate metric measures the ability of the system to continuously and reliably deliver requested files to the users. Since finding the optimal policy is very complex in general,  
we restrict our study to a specific class of policies given by the following mild assumptions: 
\begin{definition}[Admissible class policies $\Pi^{CC}$]
	The admissible policies have the following characteristics: 
	\begin{enumerate}
		\item The caching placement and delivery follow the decentralized scheme \cite{maddah2013decentralized}.
		\item The users request distinct files, i.e. the IDs of the requested files of any two users are different.
	\end{enumerate}
\end{definition}
Since we restrict our action space, the feasibility rate region, denoted by $\Lambda^{CC}$, under the class of policies $\Pi^{CC}$ is smaller than the one for the original problem $\Lambda$. However, the joint design of caching and online delivery appears to be a very hard problem; note that the design of an optimal code for coded caching alone (i.e. without considerations of scheduling, fading channels, and fairness) is an open problem and the proposed solutions are constant factor approximations. Restricting the caching strategy to the decentralized scheme proposed in \cite{maddah2013decentralized} makes the problem amenable to analysis and extraction of conclusions for general cases such as the general setup where users may not have the symmetrical rates. For the assumption of two users requesting the same file simultaneously, it would have been more efficient to handle exceptionally the transmissions as naive broadcasting instead of using the decentralized coded caching scheme, yielding a small efficiency benefit but compounding further the problem. Note, however, the probability that two users simultaneously request the same parts of video is very low in practice, hence to simplify our model we exclude this consideration altogether. {We also point out here the work in \cite{Yu2018_exactTradeoff}, where the authors  proposed a simple scheme to deal with the case where multiple users request the same file. This scheme improves has been considered also for the  Gaussian Broadcast Channel in \cite{Amiri2017}. Although it is therefore possible to relax Assumption 2 using the delivery scheme of \cite{Yu2018_exactTradeoff, Amiri2017}, this is out of the scope of our paper.}

Our objective is to solve the \emph{fair file delivery} problem:
\begin{align}\label{eq:problem}
\overline{\boldsymbol{r}}^* = &\arg\max_{\overline{\boldsymbol{r}} \in \Lambda^{CC}}\sum_{k=1}^Kg_k(\overline{r}_k),
\end{align}

where the utility function corresponds to the \emph{alpha fair} family of  concave functions obtained by choosing: 
{\begin{align}
	g_k(x) = \begin{cases}
	w_k\frac{(d+x)^{1-\alpha}}{1-\alpha}, \alpha\neq 1\\
	w_k\log(1+x/d), \alpha = 1
	\end{cases}
	\end{align}
	for $w_k\in [0, 1]$ and }some arbitrarily small $d>0$ (used to extend the domain of the functions to $x=0$).  Tuning the value of $\alpha$ changes the shape of the utility function and consequently drives the system performance $\overline{\boldsymbol{r}}^*$ to different operating points: {Let us first assume that $w_k = 1,\forall k$. Then:}(i) $\alpha=0$ yields max sum delivery rate, (ii)  $\alpha\to\infty$ yields max-min delivery rate \cite{mowalrand}, (iii)   $\alpha = 1$ yields proportionally fair delivery rate \cite{pfscheduling}.
Choosing $\alpha\in (0,1)$ leads to a tradeoff between max sum and proportionally fair delivery rates. {We can achieve all  other points in the region by using unequal values for the weights $w_k$.}


The optimization  \eqref{eq:problem}  is designed to allow us tweak the performance of the system;  we highlight its importance by an example. 
Suppose that for a 2-user system $\Lambda$ is given by the convex set shown on Fig. ~\ref{fig:example}.
Different boundary points are obtained as solutions to \eqref{eq:problem}. {Let $w_k=1,  \forall k$.}
If we choose $\alpha=0$, the system is operated at the point that maximizes the sum $\overline{r}_1+\overline{r}_2$. The choice $\alpha\to\infty$ leads to the maximum $r$ such that $\overline{r}_1=\overline{r}_2=r$, while $\alpha=1$ maximizes the sum of logarithms. 
The  operation point A is obtained when we always broadcast to all users at the weakest user rate and use \cite{maddah2013fundamental} for coded caching transmissions.
Note that this results in a significant loss of efficiency due to the variations of the fading channel, and consequently A lies in the interior of $\Lambda$.
To reach the boundary point that corresponds to $\alpha\to\infty$ we need to carefully group users together with good instantaneous channel quality but also serve users with poor average channel quality. This shows the necessity of our approach when using coded caching in realistic wireless channel conditions.

\begin{figure}
	\begin{center}
		\includegraphics[width=0.25\textwidth,clip=]{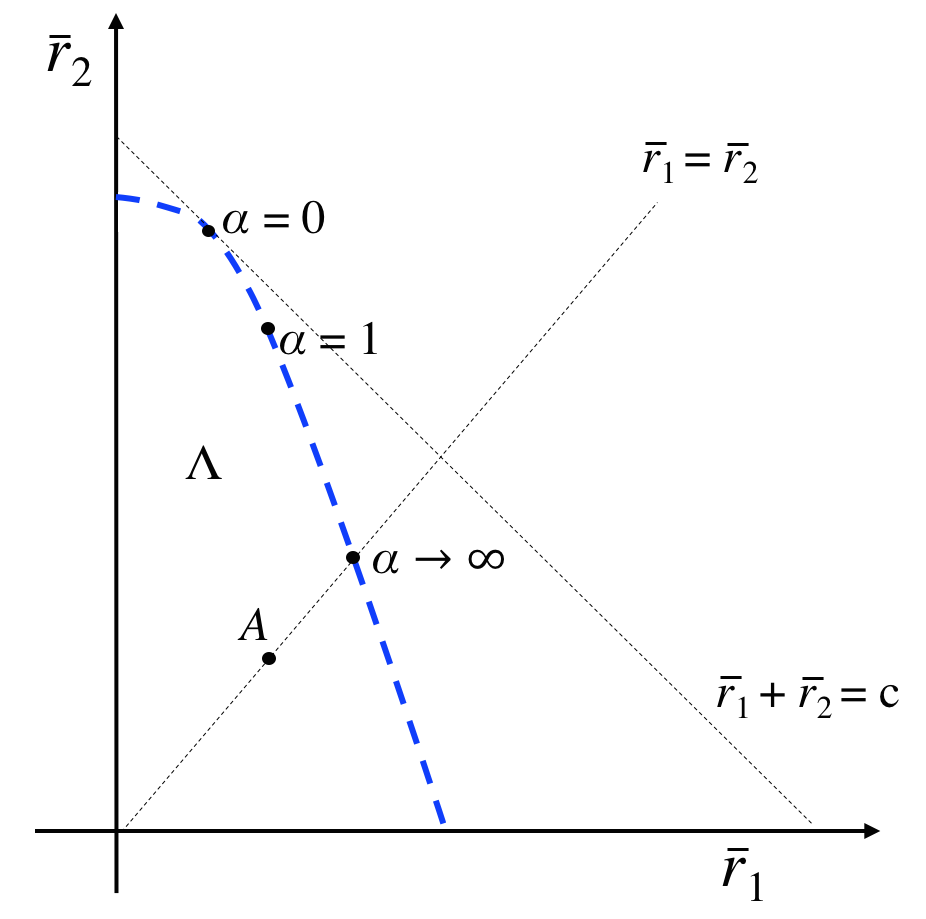}
		\caption{Illustration of the feasibility region and different performance operating points for $K=2$ users. Point A corresponds to a naive adaptation of \cite{maddah2013fundamental} on our channel model, while the rest points are solutions to our fair delivery problem. }
		\label{fig:example}
	\end{center}
\end{figure}

\subsection{Feasibility Region} 
In order to proceed, we will characterize the set of feasible file delivery rates via characterizing the stability performance of the queueing system os Section \ref{subsec:queueingStructure}.  To this end, let $\overline{a}_k = \limsup\limits_{t\rightarrow\infty}\frac{1}{t}\sum_{t=0}^{t-1}\mathbb{E}\left[a_k(t)\right],$
denote the time average number of admitted files for user $k$. We use the following definition of stability: 

\begin{definition}[Stability]\label{ef:stability} 
	A queue $S(t)$ is said to be (strongly) stable if
	\vspace{-0.1in}
	\[
	\limsup\limits_{T\rightarrow\infty}\frac{1}{T}\sum_{t=0}^{T-1}\mathbb{E}\left[S(t)\right] < \infty.
	\]
	A queueing system is said to be stable if all its queues are stable. Moreover, the stability region of a system is the set of all vectors of admitted file rates such that the system is stable.   
\end{definition}

Note that the expectation above is over fading gains, arrival processes and possibly random service processes of the queues up to time slot $t$. If the queueing system we have introduced is stable the rate of admitted files (input rate) is equal to the rate of successfully decoded files (output rate), hence we can characterize the system performance by means of the stability region of our queueing system. We let $\Gamma(\hv)$ denote the capacity region for a fixed channel state $\mathbf{h}$, as defined in Theorem \ref{th:capacityDegradedBC}. Then we have the following: 

\begin{theorem}[Stability region]\label{th:feasibilityRegion}
	Define the set $\Gamma^{CC}$ to be \[
	\Gamma^{CC} = \bigg\{  \bar{\av} \in R_+^K \big| \eqref{eq:all_packets_get_combined}, \eqref{eq:all_codewords_get_transmitted},
	\overline{\muv} \in \sum_{\hv}\phi_{\hv}\Gamma(\hv),\overline{\sigmav} \in [0,\sigma_{max}]^{2^K-1}     \bigg\},
	\]
	where 
	\begin{align}\label{eq:all_packets_get_combined}
	\sum_{\Jc: k\in \Jc}\overline{\sigma}_{\Jc} & \geq \overline{a}_k, \forall k =1,\dots,K \\ \label{eq:all_codewords_get_transmitted} 
	T_{\rm slot}\overline{\mu}_{\Ic} &\geq \sum_{\Jc: \Ic \subseteq \Jc} b_{\Jc, \Ic}\overline{\sigma}_{\Jc}, 
	\forall \Ic \subseteq \{1,2,...,K\}
	.	\end{align}
	Then, the stability region of the system is the interior of $\Gamma^{CC}$, where the above inequalities are strict. 
\end{theorem}
Constraint \eqref{eq:all_packets_get_combined} says that the aggregate service rate is greater than the arrival rate, while \eqref{eq:all_codewords_get_transmitted} implies that the long-term average rate for the subset $\Jc$ is greater than the arrival rate of the codewords intended to this subset. In terms of the queueing system defined, these constraints impose that the service rates of each queue should be greater than their arrival rates, thus rendering them stable \footnote{We restrict  vectors $\overline{\av}$  to the interior of $\Gamma^{CC}$, since arrival rates at the boundary are exceptional cases of no practical interest, and require special treatment.}. The proof of this theorem relies on existence of static policies, i.e. randomized policies whose decision distribution depends only on the realization of the channel state. See the Appendix, Section \ref{ssec:StaticPolicies} for a definition and results on these policies.   

We now restrict to the use of Markovian policies, i.e. policies that choose $\{\av(t), \sigmav(t), \muv(t)\}$ based only the state of the system at the beginning of time slot t, $\{\hv(t), \Sm(t), \Qm(t)\}$, and not the time index itself. Since the channel process $\hv(t)$ is a sequence of i.i.d. realizations of the channel states (the same results hold if, more generally, $\hv(t)$ is an ergodic Markov chain), we can obtain any admitted file rate vector $\overline{\av}$ in the stability region by using such policies, therefore the restriction is without loss of optimality. In addition, under any Markovian poliy, the state $(\Sm(t), \Qm(t))$ evolves as a Markov chain, which implies that our stability definition is equivalent to that Markov chain being ergodic with every queue having finite mean under the stationary distribution. We can the conclude that if we develop a policy that keeps {\it user queues} $\Sm(t)$ stable, then all admitted files will, at some point, be combined into codewords. Additionally, if {\it codeword queues} $\Qm(t)$ are stable, then all generated codewords will be successfully conveyed to their destinations. This in turn means that all receivers will be able to decode the admitted files that they requested: 

\begin{lemma}\label{lem:equivalence}
	The region of all feasible delivery rates $\Lambda^{CC}$ is the same as the stability region of the system, i.e. $\Lambda^{CC} = Int(\Gamma^{CC})$.
\end{lemma}
\begin{proof}
	Please refer to Appendix \ref{appendix:equivalence_proof}. 
\end{proof}
\noindent Lemma \ref{lem:equivalence} implies the following Corollary. 
\begin{corollary}\label{cor:equivalentOptimization}
	Solving \eqref{eq:problem} is equivalent to finding a policy $\pi$ such that 
	\begin{align}\label{eq:objective2}
	\overline{\boldsymbol{a}}^{\pi} =& \arg\max\sum_{k=1}^Kg_k(\overline{a}_k)\\ 
	\text{s.t.}\quad & \text{the system is stable.} \notag
	\end{align}
\end{corollary}

This implies that the solution to the original problem \eqref{eq:problem} in terms of the long-term average rates is equivalent to the new problem in terms of the admission rates stabilizing the system. Next Section provides a set of the explicit solutions to this new problem.  
%

\section{Proposed Online Delivery Scheme}\label{sec:proposed}

\subsection{Admission Control and Codeword Routing}

Our goal is to find a control policy that optimizes \eqref{eq:objective2}. To this aim, we need to introduce one more set of queues. These queues are virtual, in the sense that they do not hold actual file demands or bits, but are merely counters to drive the control policy. Each user $k$ is associated with a queue $U_k(t)$ which evolves as follows: 
\begin{align}\label{eq:utilityQueues}
U_k(t+1) = \left[U_k(t) - a_k(t)\right]^+ + \gamma_k(t)
\end{align}
where $\gamma_k(t)$ represents the arrival process to the virtual queue and is an additional control parameter.  We require these queues to be  stable:  The actual mean file admission rates are greater than the virtual arrival rates and the control algorithm actually seeks to optimize the time average of the virtual arrivals $\gamma_k(t)$. However, since $U_k(t)$ is stable, its service rate, which is the actual admission rate, will be greater than the rate of the virtual arrivals, therefore giving the same optimizer. Stability of all other queues will guarantee that admitted files will be actually delivered to the users. With thee considerations, $U_k(t)$ will be a control indicator such that when $U_k(t)$ is above $S_k(t)$ then we admit files into the system else we set $a_k(t)=0$. In particular, we will control the way $U_k(t)$ grows over time using the actual utility objective $g_k(.)$ such that a user with rate $x$ and rapidly increasing utility $g_k(x)$ (steep derivative at $x$) will also enjoy a rapidly increasing $U_k(t)$ and hence admit more files into the system. 

In our proposed policy, the arrival process to the virtual queues are given by 
\begin{equation}\label{eq:virtualArrivalsOpt}
\gamma_k(t) = \arg\max\limits_{0\leq x\leq \gamma_{k,\max}}\left[Vg_k(x) - U_k(t)x\right]
\end{equation}
In the above, $V>0$ is a parameter that controls the utility-delay tradeoff achieved by the algorithm (see Theorem \ref{th:optimality_infinite}). 

We present our on-off policy for admission control and routing.  
For every user $k$, admission control chooses $a_k(t)$ demands given by 
\begin{align}\label{eq:adm}
a_k(t) = \gamma_{k, \max} \onev\{U_k(t) \geq S_k(t) \}
\end{align}
For every subset $\Jc \subseteq \{1,\dots, K\}$, routing combines $\sigma_{\Jc}(t)$ demands of users in $\Jc$ given by
\begin{align}\label{eq:sig}
\sigma_{\Jc}(t) = \sigma_{\max} \onev\left\{ \sum_{k\in\Jc}S_k(t) > \sum_{\Ic: \Ic \subseteq \Jc}\frac{b_{\Jc, \Ic}}{F^2} Q_{\Ic}(t)  \right\}
.\end{align}


\subsection{Scheduling and Transmission}
In order to stabilize all {\it codeword queues}, the scheduling and resource allocation explicitly solves the following weighted sum rate maximization at each slot $t$, where the weight of the subset $\Jc$ corresponds to the queue length of $Q_{\Jc}$ 
\begin{align}\label{eq:WSRstability}
\muv(t) = \arg\max\limits_{\rv \in\Gamma(\hv(t))}\sum_{\Jc\subseteq \{1, \dots, K\}} Q_{\Jc}(t)r_{\Jc}.
\end{align}

Since there are $2^K-1$ codeword queues (one for each subset of users), the above is a potentially very complicated problem. However, by exploiting the properties of the capacity region $\Gamma(\hv(t))$, the problem can be simplified as the following result shows:  
\begin{theorem}
	Rearrange the user indices such that $h_1(t)\geq h_2(t)\geq\dots\geq h_{K}(t)$. Then, the weighted sum rate maximization with $2^K-1$ variables in \eqref{eq:WSRstability} reduces to a simpler problem with $K$ variables, given by 
	\begin{align}\label{eq:powerallocation}
	\max_{\pv} \sum_{k=1}^K \tilde{\theta}_k \log\frac{1+h_k(t) \sum_{j=1}^{k}p_j }{1+ h_k(t)\sum_{j=1}^{k-1} p_j }. 
	\end{align}
	where $\pv=(p_1,\dots, p_K)\in \RR_+^K$ is a positive real vector satisfying the total power constraint, and $\tilde{\theta}_k$ denotes the largest weight for user $k$ 
	\[
	\tilde{\theta}_k=\max_{\Kc: k\in \Kc \subseteq \{1,\dots, k\}}Q_{\Kc}(t).
	\]
\end{theorem}
\begin{proof}
	The proof builds on the simple structure of the capacity region. We remark that for a given power allocation of users $1$ to $k-1$, user $k$ sees $2^{k-1}$ messages $\{M_{\Jc}\}$ for all $\Jc$ such that $k\in \Jc \subseteq \{1,\dots, k\}$ with the equal channel gain. For a given set of $\{p_j\}_{j=1}^{k-1}$, the capacity region of these messages is a simple hyperplane 
	characterized by $2^{k-1}$ vertices $\tilde{R}_{k} \ev_i$ for $i=1, \dots, 2^{k-1}$, where $\tilde{R}_{k}$ is the sum rate of user $k$ in the RHS of \eqref{eq:capRegion} and $\ev_i$ is a vector with one for the $i$-th entry and zero for the others. Therefore, the weighted sum rate is maximized for user $k$ by selecting the vertex corresponding to the largest weight, denoted by $\tilde{\theta}$. This holds for any $k$. 
\end{proof}

We provide an efficient algorithm to solve this power allocation problem as a special case of the parallel Gaussian broadcast channel studied in \cite[Theorem 3.2]{TseOptimal}. Following \cite{TseOptimal}, we define the rate utility function for user $k$ given by 
\begin{align}
u_k(z)= \frac{\tilde{\theta}_k}{1/h_k(t)+z}-\lambda,
\end{align}
where $\lambda$ is a Lagrangian multiplier. The optimal solution corresponds to selecting 
the user with the maximum rate utility at each $z$ and the resulting power allocation for user $k$ is 
\begin{align}\label{eq:optimalalpha}
p^*_k = \left\{ z:   [\max_j u_j(z) ]_+ = u_k(z)  \right \}
\end{align}
with $\lambda$ satisfying 
\begin{align} \label{eq:195}
P=\left [ \max_k \frac{\tilde{\theta}_k}{\lambda} -\frac{1}{h_k(t)} \right]_+.
\end{align}

Algorithm 1 summarizes our online delivery scheme. 




\subsection{Practical Implementation}\label{subsection:PI}
When user requests arrive dynamically and the delivery phase is run continuously, it is not clear when and how the base station shall disseminate the useful side information to each individual users. This motivates us to consider a practical solution which associates a header to each sub-file $W_{i|\Jc}$ for $i=1, \dots, N$ and $\Jc\subseteq \{1, \dots, K\}$ . Namely, any sub-file shall indicate the following information prior to message symbols: a) the indices of files; b) the identities of users who cache (know) the sub-files \footnote{We assume here for the sake of simplicity that the overhead due to a header is negligible. This implies in practice that each of sub-files is arbitrarily large.}.

At each slot $t$, the base station knows the cache contents of all users $Z^K$, the sequence of the channel state $\hv^t$, as well as that of the demand vectors $\dv^t$. Given this information, the base station constructs and transmits either a message symbol or a header at channel use $i$ in slot $t$ as follows. 
\begin{align}
x_i(t) = \begin{cases}
f^{\rm h}_{t, i}(\dv^t, Z^K) & \text{if header }\\
f^{\rm m}_{t,i}(\{W_{d_k(\tau)}: \forall k,\tau \leq t\}, \hv^t) & \text{if message}
\end{cases}
\end{align}
where $f^{\rm h}_{t, i},f^{\rm m}_{t, i}$ denotes the header function, the message encoding function, respectively, at channel use $i$ in slot $t$.


\begin{table*}[ht]
	\caption{Parameters}
	\label{tab:1}
	\begin{center}
		
		\begin{tabular}{ll}\hline
			$Q_{\Ic}(t)$ & 
			codeword queue storing XOR-packets intended users in $\Ic$. \\
			
			$S_{k}(t)$& user queue storing admitted files for user $k$. \\
			$U_{k}(t)$& virtual queue for the admission control. \\
			\hline
			$\sigma_{\Jc}(t)$ & decision variable of number of combined requests for users $\Jc$ in $[0,\sigma_{\max}]$.\\
			$\mu_{\Ic}(t)$ & decision variable for multicast transmission rate to users $\Ic$.\\
			
			$a_{k}(t)$& decision variable of the number of admitted files for user $k$ in $[0,\gamma_{\max}]$. \\
			$\gamma_k(t)$& the arrival process to the virtual queue in $[0,\gamma_{\max}]$, given by eq. \eqref{eq:virtualArrivalsOpt}. \\
			\hline
			$Z_{k}$ & cache content for user $k$\\
			$\overline{D}_k(t)$ & number of successfully decoded files by user $k$ up to slot $t$.\\ 
			$\overline{A}_k(t)$ & number of (accumulated) requested files by user k up to slot $t$.\\ 
			$\overline{r}_k$ & time average delivery rate equal to $\liminf_{t\rightarrow\infty} \frac{\overline{D}_k(t)}{t}$ in files/slot.\\
			$\lambda_{k}$ & mean of the arrival process.\\
			$b_{\Jc,\Ic}$ & length of codeword intended to users $\Ic$ from applying coded caching for user in $\Jc$. \\
			$T_{\rm slot}$ & number of channel use per slot.\\
			$\Gamma(\hv)$ & the capacity region for a fixed channel state $\hv$.\\
			$\Hc$ & the set of all possible channel states.\\
			$\phi_{\hv}$ & the probability that the channel state at slot $t$ is $\hv\in \Hc$.\\
			\hline
		\end{tabular}
	\end{center}
\end{table*}

\begin{algorithm}
	\begin{algorithmic}[1]\label{algo}
		\State \textbf{PLACEMENT (same as \cite{maddah2013decentralized}):}
		
		\State Fill the cache of each user $k$
		
		$Z_k =\{ W_{i \cond \Jc}:  \;\; \Jc \subseteq\{1,\dots,K\}, k\in\Jc , \forall i =1,\dots, N \}.$
		
		\State \textbf{DELIVERY:} 
		\State \textbf{for} $t = 1,\dots,T$
		
		\State Decide the arrival process to the virtual queues 
		
		$\gamma_k(t) = \arg\max\limits_{0\leq x\leq \gamma_{k,\max}}\left[Vg_k(x) - U_k(t)x\right]$
		
		\State  Decide the number of admitted files

		$a_k(t) = \gamma_{k, \max} \onev\{U_k(t) \geq S_k(t) \}$ . 
		\State Update the virtual queues
		
		$U_k(t+1) = \left[U_k(t) - a_k(t)\right]^+ + \gamma_k(t)$
		\State Decide the number of files to be combined 
		\begin{flalign}
		\sigma_{\Jc}(t) = \sigma_{\max} \onev\left\{ \sum_{k\in\Jc}S_k(t) > \sum_{\Ic: \Ic \subseteq \Jc}\frac{b_{\Jc, \Ic}}{F^2} Q_{\Ic}(t)  \right\}.\nonumber
		\end{flalign}
		\State Scheduling decides the instantaneous rate

		$\muv(t) = \arg\max\limits_{\rv \in\Gamma(\hv(t))}\sum_{\Jc\subseteq \{1, \dots, K\}} Q_{\Jc}(t)r_{\Jc}.$
		\State Update user queues and codeword queues:
		
		$S_k(t+1) = \left[S_k(t) - \sum_{\Jc: k\in \Jc}\sigma_{\Jc}(t)\right]^+ + a_k(t) $,
		\[Q_{\Ic}(t+1) =  \left[Q_{\Ic}(t) - T_{\rm slot}\mu_{\Ic}(t)\right]^+  + \sum_{\Jc:\Ic\subseteq\Jc}b_{\Jc,\Ic}\sigma_{\Jc}(t).\]
		
	\end{algorithmic}
	\caption{Proposed delivery scheme}
	\label{alg1}
\end{algorithm}

\begin{table*}[ht]
	\vspace{-10pt}
	\caption{Codeword queues evolution for $\mu_{\{1,2\}}(t)>0$, $\mu_{\{1,2,3\}}(t)>0$ and $\sigma_{\{1,2\}}(t)=\sigma_{\{1\}}(t)=1$.}
	\label{tab:2}
	\begin{center}
		\begin{tabular}{|c|c|c|c|c|c|c|c|}\hline
			& $\Qc_{\{1\}}$ & $\Qc_{\{2\}}$ & $\Qc_{\{3\}}$& $\Qc_{\{1,2\}}$ & $\Qc_{\{1,3\}}$ & $\Qc_{\{2,3\}}$ & $\Qc_{\{1,2,3\}}$ \\
			\hline
			$\Qc_{\Jc}(t)$ & $A_{\emptyset}$ & $B_{\emptyset}$ & $C_{\emptyset}$ & $A_{2}\oplus B_{1}$ &  $A_{3}\oplus C_{1}$  &  $B_{3}\oplus C_{2}$ &  $A_{23}\oplus B_{13}\oplus C_{12}$\\
			\hline
			\makecell{Output\\ $\mu_{\{1,2\}}(t)>0$, $\mu_{\{1,2,3\}}(t)>0$}   & - & - & - & $A_{2}\oplus B_{1}$ &  -  &  - &  $A_{23}\oplus B_{13}\oplus C_{12}$\\
			\hline
			\makecell{ Input\\ $\sigma_{\{1,2\}}(t)=\sigma_{\{1\}}(t)=1$} & \makecell{$D_{\emptyset}; D_{3}$ \\ $\{F_{\Jc}\}_{1\notin \Jc}$ } &  $E_{\emptyset}$; $E_{3}$  & - & \makecell{$E_{1}\oplus D_{2}$\\ $E_{13}\oplus D_{23}$ } &  -  &  - &  - \\
			\hline
			$\Qc_{\Jc}(t+1)$ & \makecell{$A_{\emptyset}$; $D_{\emptyset}$; $D_{3}$ \\ $\{F_{\Jc}\}_{1\notin \Jc}$ } & \makecell{$B_{\emptyset}$\\ $E_{\emptyset}$; $E_{3}$} & $C_{\emptyset}$ & \makecell{$E_{1}\oplus D_{2}$\\ $E_{13}\oplus D_{23}$ } &  $A_{3}\oplus C_{1}$  &  $B_{3}\oplus C_{2}$ &  -\\
			\hline
			
		\end{tabular}
	\end{center}
\end{table*}
\begin{example} 
	We conclude this section by providing an example of a slot for our online delivery scheme for $K=3$ users as illustrated in Fig.~ \ref{fig:queueing_system_real}. 
	
	We focus on the evolution of codeword queues between two slots, $t$ and $t+1$. The exact backlog of codeword queues is shown in Table \ref{tab:2}. Given the routing and scheduling decisions ($\sigma_{\Jc}(t)$ and $\mu_{\Jc}(t)$), we provide the new states of the queues at the next slot in the same Table.
	
	We suppose that $h_1(t)>h_2(t)>h_3(t)$. The scheduler uses \eqref{eq:WSRstability} to allocate positive rates to user set $\{1,2\}$ and $\{1,2,3\}$ given by $\mu_{\{1,2\}}, \mu_{\{1,2,3\}}$ and multicasts the superposed signal $x(t)=B_{\emptyset}+ B_{3}\oplus C_{2}$. User $3$ decodes only $B_{3}\oplus C_{2}$.  User $2$ decodes first $B_{3}\oplus C_{2}$, then subtracts it and decodes  $B_{\emptyset}$. Note that the sub-file $B_{\emptyset}$ is simply a fraction of the file $B$ whereas the sub-file $B_{3}\oplus C_{2}$ is a linear combination of two fractions of different files. In order to differentiate between each sub-file, each user uses the data information header existing in the received signal. In the next slot, the received sub-files are evacuated from the codeword queues.
	
	For the routing decision, the server decides at slot $t$ to combine $D$ requested by user $1$ with $E$ requested by user $2$ and to process $F$ requested by user $1$ uncoded. Therefore, we have $\sigma_{\{1,2\}}(t)=\sigma_{\{1\}}(t)=1$ and $\sigma_{\Jc}(t)=0$ otherwise. Given this codeword construction, codeword queues have inputs that change its state in the next slot as described in Table \ref{tab:2}. 
\end{example}

\subsection{Performance Analysis}

Here  we present the main result of the paper, by proving that our proposed online algorithm achieves near-optimal performance for all policies within the class $\Pi^{CC}$:

\begin{theorem}\label{th:optimality_infinite}
	Let $\overline{r}^{\pi}_k$ the mean time-average delivery rate for user $k$  achieved by the proposed policy. Then 
	\begin{align}\nonumber
	\sum_{k=1}^Kg_k(\overline{r}^{\pi}_k) &\geq \max_{\overline{\rv} \in\Lambda^{CC}}\sum_{k=1}^Kg_k(\overline{r}_k) - \frac{B}{V} \\ \nonumber
	\limsup\limits_{T\rightarrow \infty}\frac{1}{T}\sum_{t=0}^{T-1}\mathbb{E}\left\{\hat{Q}(t)\right\}  &\leq \frac{B + V\sum_{k=1}^Kg_k(\gamma_{max,k})}{\epsilon_0},
	\end{align}
	where $\hat{Q}(t)$ is the sum of all queue lengths at the beginning of time slot $t$, thus a measure of the mean delay of file delivery. The quantities $B$ an $\epsilon_0$ are constants that depend on the statistics of the system and are given in equations \eqref{eq:Beta} and \eqref{eq:epsilon1}-\eqref{eq:epsilon2} of the Appendix, respectively. 
\end{theorem}

The above theorem states that, by tuning the constant $V$, the utility resulting from our online policy can be arbitrarily close to the optimal one, where there is a tradeoff between the guaranteed optimality gap $\bigO{1/V}$ and the upper bound on the total buffer length $\bigO{V}$. We note that these tradeoffs are in direct analogue to the converge error vs step size of the subgradient method in convex optimization.

\begin{proof}[Sketch of proof]
	For proving the Theorem, we use the Lyapunov function 
	\[
	L(t) = \frac{1}{2}\left(\sum_{k=1}^KU_k^2(t) + S_k^2(t) + \sum_{\Ic \in 2^{\Kc}}\frac{1}{F^2}Q_{\Ic}^2(t)\right)
	\]
	and specifically the related drift-plus-penalty quantity, defined as: 
	\[\mathbb{E}\left\{L(t+1) - L(t)| \mathbf{S}(t), \mathbf{Q}(t), \mathbf{U}(t)\right\} - V\mathbb{E}\left\{\sum_{k=1}^Kg(\gamma_k(t))|\mathbf{S}(t), \mathbf{Q}(t), \mathbf{U}(t) \right\}\]. The proposed algorithm is such that it minimizes (a bound on) this quantity. The main idea is to use this fact in order to compare the evolution of the drift-plus-penalty under our policy and two "static" policies, that is policies that take random actions (admissions, demand combinations and wireless transmissions), drawn from a specific distribution, based only on the channel realizations (and knowledge of the channel statistics). We can prove from Theorem 4 that these policies can attain every feasible delivery rate. The first static policy is one such that it achieves the stability of the system for  an arrival rate vector $\av'$ such that $\av'+\deltav \in \partial\Lambda^{CC}$. Comparing with our policy, we deduce strong stability of all queues and the bounds on the queue lengths by using a Foster-Lyapunov type of criterion. In order to prove near-optimality, we consider a static policy that admits file requests at rates $\av^* = \arg\max_{\av}\sum_kg_k(a_k)$ and keeps the queues stable in a weaker sense (since the arrival rate is now in the boundary $\Lambda^{CC}$). By comparing the drift-plus-penalty quantities and using telescopic sums and Jensen's inequality on the time average utilities, we obtain the near-optimality of our proposed policy.  
	
	The full proof, as well as the expressions for the constants $B$ and $\epsilon_0$, are in Section \ref{appendix:performance_proof} of the Appendix (equations \eqref{eq:Beta} and \eqref{eq:epsilon1} - \eqref{eq:epsilon2}, respectively).  
\end{proof}

\section{Dynamic File Requests}\label{sec:dynamic}
In this Section, we extend our algorithm to the case where there is no infinite amount of demands for each user, rather each user requests a finite number of files at slot $t$. Let $A_k(t)$ be the number of files requested by user $k$ at the beginning of slot $t$. We assume it is an i.i.d. random process with mean $\lambda_k$ and such that $A_k(t) \leq A_{\max}$ almost surely. \footnote{The assumptions can be relaxed to arrivals being ergodic Markov chains with finite second moment under the stationary distribution} In this case, the alpha fair delivery problem is to find a delivery rate $\overline{\mathbf{r}}$ that solves  
\begin{align} \nonumber 
\text{Maximize }& \sum_{k=1}^Kg_k(\overline{r}_k) \\ \nonumber
\text{s.t. }& \overline{\mathbf{r}} \in \Lambda^{CC} \\  \nonumber
& \overline{r}_k \leq \lambda_k ,\forall k\in\{1,...,K\}
,\end{align}
where the additional constraints $\overline{r}_k \leq \lambda_k$ denote that a user cannot receive more files than the ones actually requested. 

The fact that file demands are not infinite and come as a stochastic process is dealt with by introducing one "reservoir queue" per user, $L_k(t)$, which stores the file demands that have not been admitted, and an additional control decision on how many demands to reject permanently from the system, $d_k(t)$. At slot $t$, no more demands then the ones that arrived at the beginning of this slot and the ones waiting in the reservoir queues can be admitted, therefore the admission control must have the additional constraint  
\[
a_k(t) \leq A_k(t) + L_k(t), \forall k, t
,\]
and a similar restriction holds for the number of rejected files from the system, $d_k(t)$. The reservoir queues then evolve as 
\[
L_k(t+1) = L_k(t) + A_k(t) - a_k(t) - d_k(t).
\]
The above modification with the reservoir queues has only an impact that further constrains the admission control of files to the system. The queuing system remains the same as described in Section \ref{sec:proposed}, with the user queues $\Sm(t)$, the codeword queues $\Qm(t)$ and the virtual queues $\Um(t)$. Similar to the case with infinite demands we can restrict ourselves to policies that are functions only of the system state at time slot $t$, $\{\Sm(t), \Qm(t), \Lm(t), \Am(t), \hv(t), \Um(t)\}$  without loss of optimality. Furthermore, we can show that the alpha fair optimization problem  equivalent to the problem of controlling the admission rate. That is, we want to find a policy $\pi$ such that 
\begin{align} \nonumber
&\overline{\mathbf{a}}^{\pi} = \arg\max\sum_{k=1}^Kg_k(\overline{a}_k) \\ \nonumber
\text{s.t. } & \text{the queues $\left(\Sm(t),\Qm(t), \Um(t)\right)$ are strongly stable}\\ \nonumber
& a_k(t) \leq \min[a_{max,k},  L_k(t) + A_k(t)], \forall t\geq 0, \forall k
\end{align}  

The rules for scheduling, codeword generation, virtual queue arrivals and queue updating remain the same as in the case of infinite demands in subsections  C and  D of Sec. \ref{sec:proposed}.  The only difference is that there are multiple possibilities for the admission control; see \cite{Li05} and Chapter 5 of \cite{Georgiadis06} for more details. Here we propose that at each slot $t$, any demand that is not admitted get rejected (i.e. the reservoir queues hold no demands), the admission rule is
\begin{equation}
a^{\pi}_k(t) = A_k(t) \ind{U_k(t) \geq S_k(t)},
\end{equation} 
and the constants are set as $\gamma_{k,\max}, \sigma_{\max} \geq A_{\max}$.  Using the same ideas employed in the performance analysis of the case with infinite demands and the ones employed in \cite{Li05}, we can prove that the $O(1/V) - O(V)$ utility-queue length tradeoff  of Theorem \ref{th:optimality_infinite} holds for the case of dynamic arrivals as well.

\section{Numerical Examples}

In this section, we compare our proposed delivery scheme with two other schemes described below, all building on the decentralized cache placement described in \eqref{eq:Zk} and \eqref{eq:LLN}. 

\begin{itemize}
	\item  \textbf{Our proposed scheme}: We apply Algorithm 1 for $t=10^5$ slots. Using the scheduler \eqref{eq:WSRstability}, we calculate $\mu_{\Jc}(\tau) $ denoting the rate allocated to a user set $\Jc$ at slot $\tau\leq t$. As defined in \eqref{eq:feasiblerate}, the long-term average rate of user $k$ measured in file/slot is given by
	\begin{align}\nonumber
	\overline{r}_k= \frac{T_{\rm slot}\lim_{t\rightarrow\infty} \frac{1}{t}\sum_{\tau=1}^t \sum_{\Jc: k\in \Jc} \mu_{\Jc}(\tau) }{(1-m)F}. 
	\end{align}
	Notice that the numerator corresponds to the average number of useful bits received over a slot by user $k$ and
	the denominator $(1-m)F$ corresponds to the number of bits necessary to recover one file.

	\item \textbf{Unicast opportunistic
		scheduling}: For any request, the server sends the remaining $(1-m)F$ bits to the corresponding user without combining any files. Here we only exploit the local caching gain. In each slot the transmitter sends with full power to the following user
	\[ k^{*}(t)=\arg\max_{k}\frac{\log\left( 1+h_k(t)P\right) }{ T_k(t) ^{\alpha}},
	\]
	where $T_k(t)=\frac{\sum_{1\leq\tau\leq t-1} \mu_{k}(\tau)}{(t-1)}$ is the empirical average rate for user $k$ up to slot $t$. The resulting long-term average rate of user $k$ measured in file/slot is given by 
	\[\overline{r}_k= \frac{T_{\rm slot} \lim_{t\rightarrow\infty} \frac{1}{t}\sum_{\tau=1}^t \log(1+P h_k (\tau)) \onev\{k =k^*(\tau)\}}{(1-m)F}.\]
	
	\item \textbf{Standard coded caching}: 
	We use decentralized coded caching among all $K$ users. For the delivery, non-opportunistic TDMA transmission is used. The server sends a sequence of codewords $\{V_{\Jc}\}$ at the worst transmission rate. The number of packets to be multicast in order to satisfy one demand for each user is given by \cite{maddah2013decentralized}

	\begin{align}\nonumber 
	T_{\rm tot}(K,m) = \frac{1}{m} \left(1-m\right) \left\{1-\left(1-m\right) ^K \right \}. 
	\end{align}
	Thus the average delivery rate (in file per slot) is symmetric, and given as the following
	\begin{align}\nonumber 
	\overline{r}_k= \frac{T_{\rm slot}}{T_{\rm tot}(K,m)F}\EE\left[\log(1+P\min_{i\in\{1,\dots,K\}}h_i)\right].  
	\end{align}

\end{itemize}


We consider a system with  normalized memory of $m=0.6$, power constraint $P=10dB$, file size $F=10^3$ bits and number of channel uses per slot $T_{\rm slot}=10^2$. The channel coefficient $h_k(t)$ follows an exponential distribution with mean $\beta_k$.

We compare the three algorithms for the cases where the objective of the system is sum rate maximization ($\alpha=0$) and proportional fairness ($\alpha=1$) in two different scenarios. The results depicted in Fig.~\ref{fig:snr} consider a deterministic channel with two classes of users of $K/2$ each: strong users with $\beta_k=1$ and weak users with $\beta_k=0.2$. For Fig.~\ref{fig:fading}, we consider a symmetric block fading channel with $\beta_k=1$ for all users. Finally, for Fig.~\ref{fig:asymmetricFading}, we consider a system with block fading channel and two classes of users: $K/2$ strong users with $\beta_k=1$ and $K/2$ weak users with $\beta_k=0.2$.

\begin{figure*}
	\centering
	\subfloat[Rate ($\alpha=0$) vs $K$.]{
		\includegraphics[width=.4\linewidth]{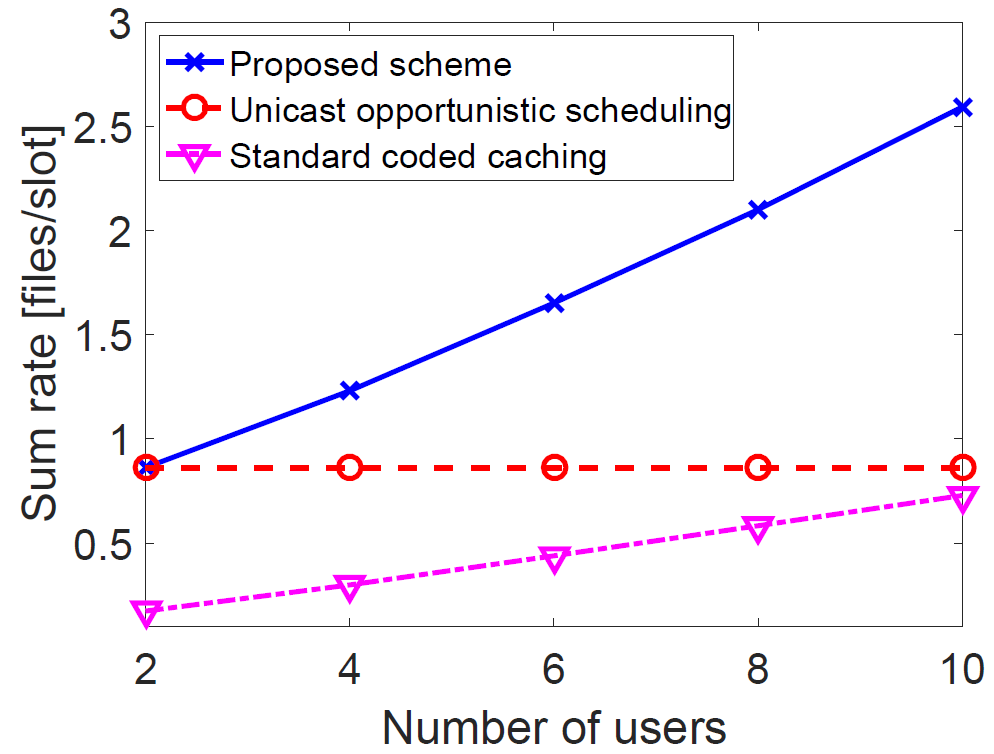}}
	\hfil
	\subfloat[Proportional fair utility ($\alpha=1$) vs $K$.]{
		\includegraphics[width=.4\linewidth]{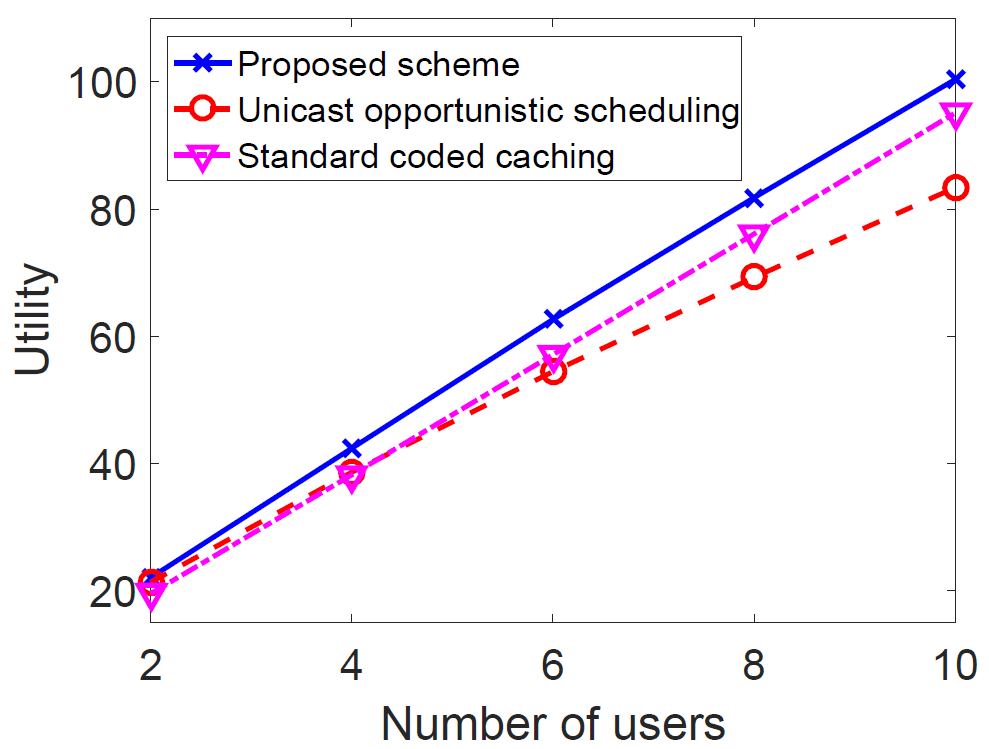}}
	\caption{Deterministic channel with different SNR.}
	\label{fig:snr}
\end{figure*}

\begin{figure*}
	\centering
	\subfloat[Rate ($\alpha=0$) vs $K$.]{
		\includegraphics[width=.4\linewidth]{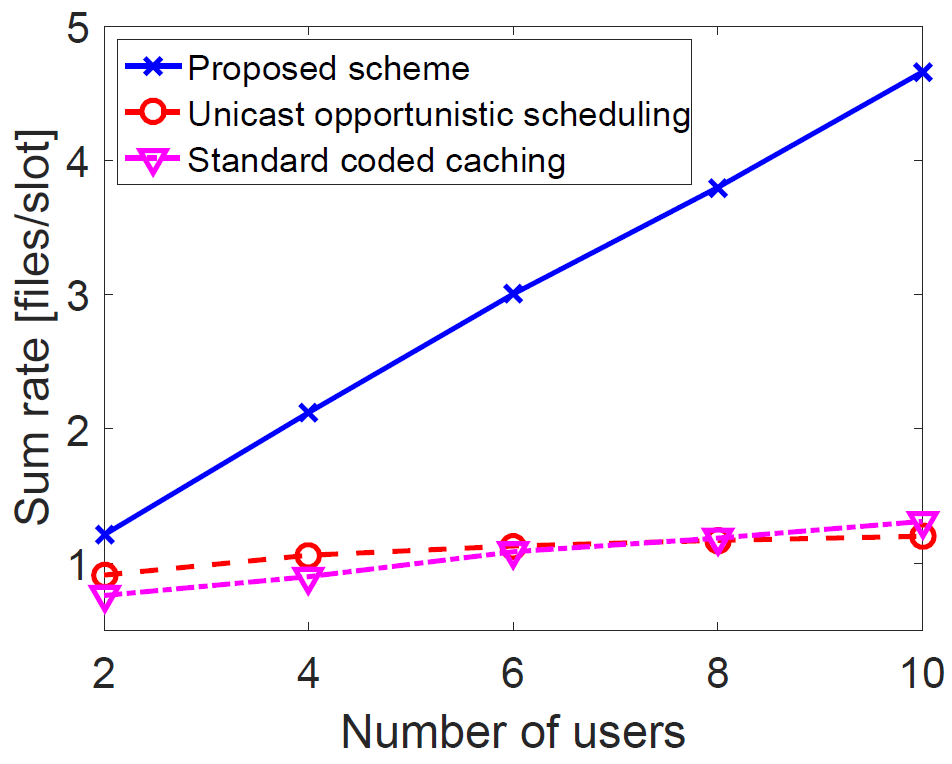}}
	\hfil
	\subfloat[Proportional fair utility ($\alpha=1$) vs $K$.]{
		\includegraphics[width=.4\linewidth]{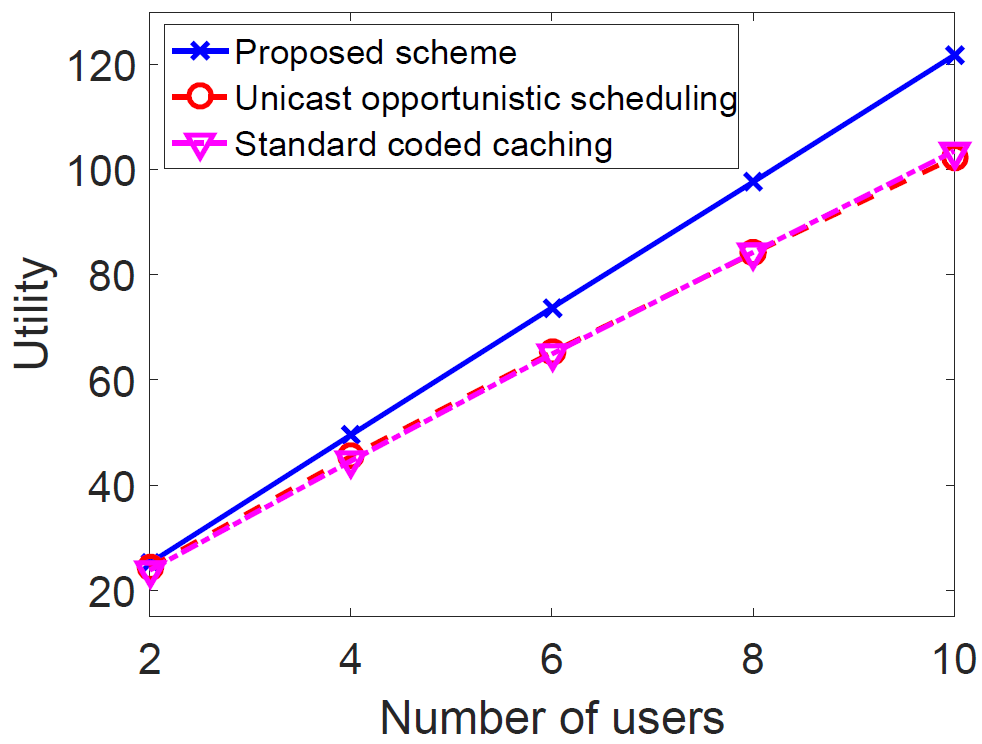}}
	\caption{Symmetric fading channel.}
	\label{fig:fading}
\end{figure*}

\begin{figure*}
	\centering
	\subfloat[Rate ($\alpha=0$) vs $K$.]{
		\includegraphics[width=.4\linewidth]{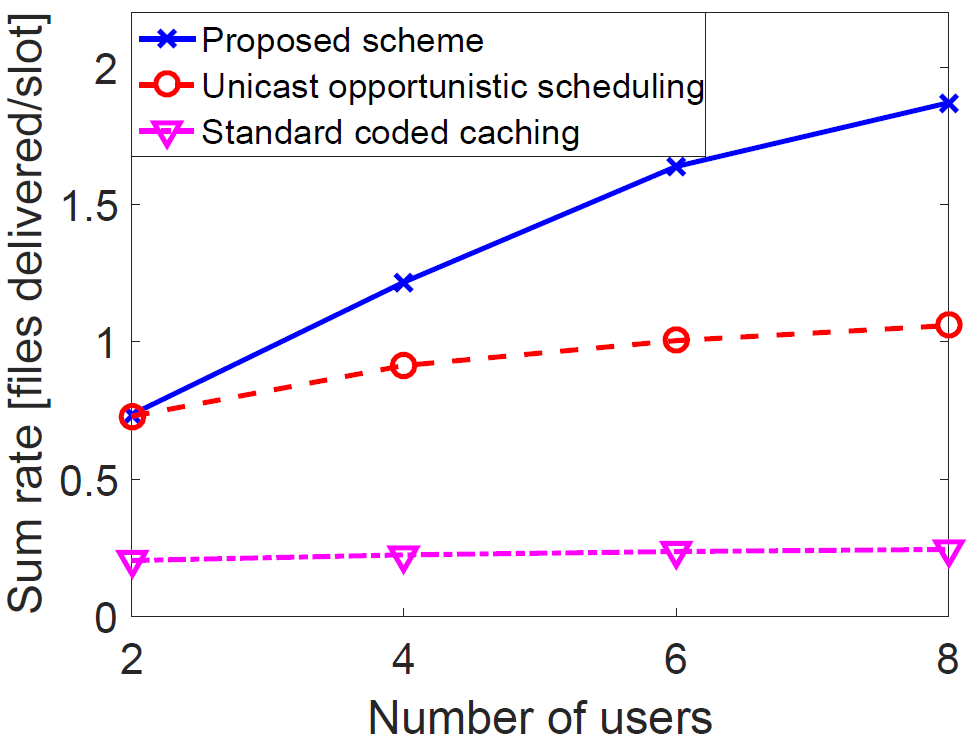}}
	\hfil
	\subfloat[Proportional fair utility ($\alpha=1$) vs $K$.]{
		\includegraphics[width=.4\linewidth]{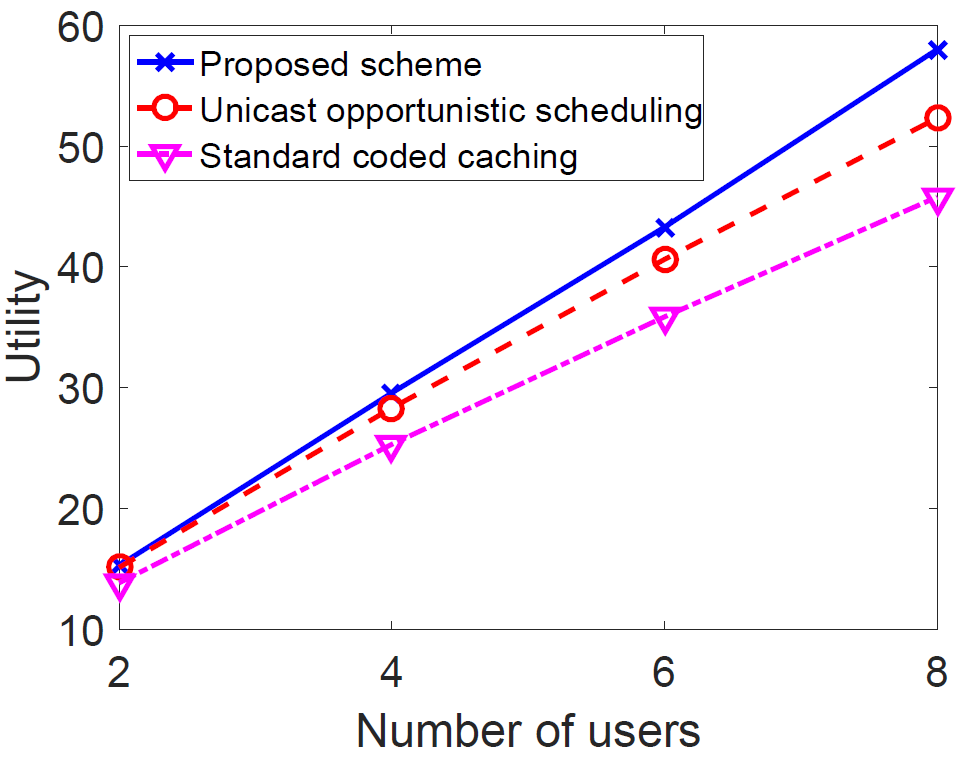}}
	\caption{Fading channels with two groups of users, each with different average SNR.}
	\label{fig:asymmetricFading}
\end{figure*}

It is notable that our proposed scheme outperforms the unicast opportunistic scheme, which maximizes the sum rate if only private information packets are to be conveyed, and standard coded caching which transmit multicast packets with the worst user channel quality. 

In Fig. \ref{fig:snr} for the deterministic channel scenario with $\alpha=0$, the unicast opportunistic scheme serves only the $K/2$ strong users in TDMA at a constant rate equal to $\frac{T_{\rm slot}}{F}\frac{\log(1+P)}{1-m}=0.865$ in file/slot. For the standard coded caching, the sum rate increases linearly with the number of users. This is because the multicast rate is constant over the deterministic channel and the behavior of $T_{\rm tot}(K,m)$ is almost constant with $m=0.6$, which makes the per-user rate almost constant. 

For the symmetric fading channel in Fig \ref{fig:fading}, the performance of unicast opportunistic and that of standard coded caching schemes are limited due to the lack of global caching gain and vanishing multicast rate, respectively.

Finally, for the case of both fading and differences in the mean SNRs, we can see from Fig.~\ref{fig:asymmetricFading} that, again, our proposed scheme outperforms the unicast opportunistic scheduling and standard coded caching both in terms of sum rate and in terms of proportional fair utility. 

In all scenarios, the relative merit of our scheme increases as the number of users grows. This can be attributed to the fact that our scheme can exploit any available multicast opportunities. \emph{Our result here implies that, in realistic wireless systems, coded caching can indeed provide a significant throughput increase when an appropriate joint design of routing and opportunistic transmission is used.} Regarding the proportional fair objective, we can see that the average sum utility increases with a system dimension for three schemes although our proposed scheme provides a  gain compared to the two others.


\section{Conclusions}\label{sec:conclusions}

In this paper, we studied coded caching over wireless fading channels in order to address its limitation governed by the user with the worst fading state. By formulating an alpha-fair optimization problem with respect to the long-term average delivery rates, we proposed a novel queueing structure that allowed us to obtain an optimal algorithm for joint file admission control, codeword construction and wireless transmissions. The main conclusion is that, by appropriately combining the multicast opportunities and the opportunism due to channel fading, coded caching can lead to significant gains in wireless systems with fading. Low-complexity algorithms which retain the benefits of our approach as well as a delay-constrained delivery scheme, are left as interesting topics of future investigation.  


\section{Appendix: Proofs}

\subsection{Proof of Theorem 1}\label{appendix:superp}
Let $M_{\Jc}$ be the message for all the users in $\Jc\subseteq [K]$ and of size $2^{nR_\Jc}$. 
We first show the converse. 
It follows that the set of $2^K - 1$ independent messages $\{
M_\Jc:\ \Jc\subseteq [K],\, \Jc\ne\emptyset \}$ can be partitioned as 
\begin{align}
\bigcup_{k=1}^K \{ M_\Jc:\ k\in\Jc\subseteq[k]\}. 
\end{align}%
We can now define $K$ independent mega-messages $\tilde{M}_k := \{ M_\Jc:\ k\in\Jc\subseteq[k]\}$ with rate $\tilde{R}_k:=\sum_{\Jc:\,k\in\Jc\subseteq[k]}R_\Jc$.
Note that each mega-message~$k$ must be decoded at least by user~$k$ reliably. Thus, the $K$-tuple $(\tilde{R}_1,\ldots,\tilde{R}_K)$ must lie inside the
private-message capacity region of the $K$-user BC. Since it is a degraded BC, the capacity region is known~\cite{el2011network}, and we have
\begin{align}
\tilde{R}_k &\le \log \frac{1+h_k \sum_{j=1}^k p_j}{1+h_k \sum_{j=1}^{k-1} p_j}, \quad k=2,\ldots,K, \label{eq:tmp99}
\end{align}%
for some $p_j\ge0$ such that $\sum_{j=1}^K p_j\le P$. This establishes the converse.  

To show the achievability, it is enough to use rate-splitting. Specifically, the transmitter first assembles the original messages into
$K$ mega-messages, and then applied the standard $K$-level superposition coding~\cite{el2011network} putting the~$(k-1)$-th signal on top of the~$k$-th signal. The
$k$-th signal has average power $p_k$, $k\in[K]$. At the receivers' side, if the rate of the mega-messages are inside the private-message
capacity region of the $K$-user BC, i.e., the $K$-tuple  $(\tilde{R}_1,\ldots,\tilde{R}_K)$ satisfies \eqref{eq:tmp99}, then each user $k$ can decode the
mega-message~$k$. Since the channel is degraded, the users~1 to $k-1$ can also decode the mega-message~$k$ and extract its own message.
Specifically, each user~$j$ can obtain $M_{\Jc}$ (if $\Jc\ni j$), from the mega-message~$\tilde{M}_k$ when $k\in\Jc\subseteq [k]$. This completes the achievability proof. 

\subsection{Static policies}\label{ssec:StaticPolicies}

An important concept for characterizing the feasibility region and proving optimality of our proposed policy is the one we will refer to here as "static policies". The concept is that decisions taken according to these policies depend only on the channel state realization (i.e. the uncontrollable part of the system) as per the following definition:  

\begin{definition}[Static Policy]
	Any policy that selects the control variables $\{\av(t), \sigmav(t), \muv(t)\}$ according to a probability distribution that depends only on the channel state $\hv(t)$ will be called a static policy. 
\end{definition} 

It is clear from the definition that all static policies belong to the set of admissible policies for our setting. An important case is where actually admission control $\av(t)$ and codeword routing $\sigmav(t)$ are decided at random and independently of everything and transmissions $\muv(t)$ are decided at by a distribution that depends only on the channel state realization of the slot: It can be shown using standard arguments in stochastic network optimization (see for example \cite{stolyar, Li05, Georgiadis06, neely10}) that the optimal long term file delivery vector and any file delivery vector in the stability region of the queueing system can be achieved by such static policies, as formalized by the following Lemmas:  

\begin{lemma}[Static Optimal Policy]\label{lem:StaticOptimalPolicy} Define a policy $\pi^*\in\Pi^{CC}$ that in each slot where the channel states are $\mathbf{h}$ works as follows:  (i) it pulls random user demands with mean $\bar{a}_k^*$, and it gives the virtual queues arrivals with mean $\overline{\gamma}_k = \overline{a}_k^*$ as well (ii) the number of combinations for subset $\Jc$ is a random variable with mean $\overline{\sigma}^*_{\Jc}$ and uniformly bounded by $\sigma_{\max}$, (iii) selects one out of $K+1$ suitably defined rate vectors $\mathbf{\muv^l}\in \Gamma(\hv), l=1,..,K+1$ with probability $\psi_{l, \hv}$. The parameters above are selected such that they solve the following problem: 
	\begin{align}\nonumber
	\max_{\overline{\av}}&\sum_{k=1}^Kg_k(\overline{a}_k^*)\\ \nonumber
	\text{s.t.}\quad & \sum_{\Jc: k\in \Jc}\overline{\sigma}^*_{\Jc} \geq \overline{a}_k^*, \forall k\in\{1,..,K\}\\ \nonumber
	& \sum_{\Jc: \Ic \subseteq \Jc}b_{\Jc, \Ic}\overline{\sigma}^*_{\Jc} \leq T_{\rm slot} \sum_{\hv}\phi_{\hv}\sum_{l=1}^{K+1}\psi_{l, \hv}\mu^l_{\Ic}(\hv), \\ \nonumber 
	& \qquad \qquad \qquad \qquad \qquad \qquad \qquad \forall \Ic \subseteq \{1,2,...,K\}
	\end{align}
	Then, $\pi^*$ results in the optimal delivery rate vector (when all possible policies are restricted to set $\Pi^{CC}$). 	
\end{lemma}

\begin{lemma}[Static Policy for the $\delta-$ interior of $\Gamma^{CC}$]\label{lem:StablePolicy} Define a policy $\pi^{\delta}\in\Pi^{CC}$ that in each slot where the channel states are $\mathbf{h}$ works as follows:  (i) it pulls random user demands with mean $\overline{a}_k^{\delta}$ such that $(\overline{\av}+\deltav)\in\Gamma^{CC}$, and gives the virtual queues random arrivals with mean $\overline{\gamma}_k \leq \overline{a}_k + 
	\epsilon'$ for some $\epsilon'>0$ (ii) the number of combinations for subset $\Jc$ is a random variable with mean $\overline{\sigma}^{\delta}_{\Jc}$ and uniformly bounded by $\sigma_{\max}$, (iii) selects one out of $K+1$ suitably defined rate vectors $\mathbf{\muv^l}\in \Gamma(\hv), l=1,..,K+1$ with probability $\psi^{\delta}_{l, \hv}$. The parameters above are selected such that: 
	\begin{align}\nonumber
	& \sum_{\Jc: k\in \Jc}\overline{\sigma^{\delta}}_{\Jc} \geq \epsilon + \overline{a}_k^{\delta}, \forall k\in\{1,..,K\}\\ \nonumber
	& \sum_{\Jc: \Ic \subseteq \Jc}b_{\Jc, \Ic}\overline{\sigma}^{\delta}_{\Jc} \leq \epsilon + T_{\rm slot}\sum_{\hv}\phi_{\hv}\sum_{l=1}^{K+1}\psi^{\delta}_{l, \hv}\mu^l_{\Ic}(\hv), \forall \Ic \in 2^{\Kc}
	\end{align}
	for some appropriate $\epsilon < \delta$. Then, the system under $\pi^{\delta}$ has mean incoming rates of $\overline{\av}^{\delta}$ and is strongly stable. 
\end{lemma}  

\subsection{Proof of Lemma \ref{lem:equivalence}}\label{appendix:equivalence_proof} 

We prove the Lemma in two parts: (i) first we prove that $\Gamma^{CC} \subseteq \Lambda^{CC}$ and (ii) then that $\left(\Gamma^{CC}\right)^c \subseteq \left(\Lambda^{CC}\right)^c$.

For the first part, we show that if $\overline{\av}\in Int(\Gamma^{CC})$ then also $\lambdav\in \Lambda^{CC}$, that is the long term file delivery rate vector observed by the users as per \eqref{eq:feasiblerate} is $\overline{\rv}=\overline{\av}$. Denote $\overline{A}_k(t)$ the number of files that have been admitted to the system for user $k$ up to slot $t$. Also, note that due to our restriction on the class of policies $\Pi^{CC}$ and our assumption about long enough blocklengths, there are no errors in decoding the files, therefore the number of files correctly decoded for user $k$ till slot $t$ is $\overline{D}_k(t)$ . From Lemma \ref{lem:StablePolicy} it follows that there exists a static policy $\pi^{RAND}$, the probabilities of which depending only on the channel state realization at each slot, for which the system is strongly stable. Since the channels are i.i.d. random with a finite state space and queues are measured in files and bits, the system now evolves as a discrete time Markov chain $(\mathbf{S}(t), \mathbf{Q}(t), \mathbf{H}(t))$, which can be checked that is aperiodic, irreducible and with a single communicating class. In that case, strong stability means that the Markov chain is ergodic with finite mean. 

Further, this means that the system reaches to the set of states where all queues are zero infinitely often. Let $T[n]$ be the number of timeslots between the $n-$th and $(n+1)-$th visit to this set (we make the convention that $T[0]$ is the time slot that this state is reached for the first time). In addition, let $\tilde{A}_k[n], \tilde{D}_k[n]$ be the number of demands that arrived and were delivered in this frame, respectively. Then, since within this frame the queues start and end empty, we have \[
\tilde{A}_k[n] = \tilde{D}_k[n], \forall n, \forall k.
\]

\noindent In addition since the Markov chain is ergodic,  
\[
\overline{a}_k = \lim\limits_{t\rightarrow\infty}\frac{{A}(t)}{t} = \lim\limits_{N\rightarrow\infty}\frac{\sum_{n=0}^N\tilde{A}_k[n]}{\sum_{n=0}^NT[n]}
\]
and 
\[
\overline{r}_k = \lim\limits_{t\rightarrow\infty}\frac{D(t)}{t} = \lim\limits_{N\rightarrow\infty}\frac{\sum_{n=0}^N\tilde{D}_k[n]}{\sum_{n=0}^NT[n]}
\]

\noindent Combining the three expressions, $\overline{\rv} = \overline{\av}$ thus the result follows. 

We now proceed to show the second part, that is given any arrival rate vector $\overline{\av}$ that is not in the stability region of the queueuing system we cannot have a long term file delivery rate vector $\overline{\rv} = \overline{\av}$. Indeed, since $\overline{\av}\notin \Gamma^{CC}$, for any possible $\overline{\sigmav}$ satisfying \eqref{eq:all_packets_get_combined}, for every $\overline{\muv}\in\sum_{\hv\in\Hc}\phi_{\hv}\Gamma(\hv)$ there will be some subset(s) of users for which the corresponding inequality \eqref{eq:all_codewords_get_transmitted} is violated. Since codeword generation decisions are assumed to be irrevocable and $\sum_{\hv\in\Hc}\phi_{\hv}\Gamma(\hv)$ is the capacity region of the wireless channel, the above implies that there is not enough wireless capacity to satisfy a long term file delivery rate vector of $\overline{\av}$. Therefore, $\overline{\av}\notin \Lambda^{CC}$, finishing the proof. \footnote{We would also need to check the boundary of $\Gamma^{CC}$. Note, however, that by similar arguments we can show that for each vector on $\partial\Gamma^{CC}$ we need to achieve a rate vector on the boundary of the capacity region of the wireless channel. Since, as mentioned in the main text, we do not consider boundaries in this work, we can discard these points.} 

\subsection{Proof ot Theorem \ref{th:optimality_infinite}}\label{appendix:performance_proof} 

We first look at static policies, which take random decisions based only on the channel realizations. We focus on two such policies: (i) one that achieves the optimal utility, as described in Lemma \ref{lem:StaticOptimalPolicy} and (ii) one that achieves (i.e. admits and stabilizes the system for that) a rate vector in the $\delta-$ interior of $\Lambda^{CC}$ (for any $\delta >0$), as described in Lemma \ref{lem:StablePolicy}. Then, we show that our proposed policy minimizes a bound on the drift of the quadratic Lyapunov function and compare with the two aforementioned policies: Comparison with the  second policy proves strong stability of the system under our proposed policy, while comparison with the first one proves almost optimality. 

From Lemma \ref{lem:equivalence} and Corollary \ref{cor:equivalentOptimization}, it suffices to prove that under the online policy the queues are strongly stable and the resulting time average admission rates maximize the desired utility function subject to minimum rate constraints. 

The proof of the performance of our proposed policy is based on applying Lyapunov optimization theory \cite{neely10} with the following as Lyapunov function (where we have defined $\mathbf{Z}(t) = (\mathbf{S}(t), \mathbf{Q}(t), \mathbf{U}(t))$ to shorten the notation)
\[
L(\Zm)= L(\mathbf{S}, \mathbf{Q}, \mathbf{U}) = \frac{1}{2}\left(\sum_{k=1}^KU_k^2(t) + S_k^2(t) + \sum_{\Ic \in 2^{\Kc}}\frac{Q_{\Ic}^2(t)}{F^2}\right).\]
We then define the drift of the aforementioned  Lyapunov function as
\[
\Delta L(\mathbf{Z}) = \mathbb{E}\left\{L(\mathbf{Z}(t+1))-L(\mathbf{Z}(t))|\mathbf{Z}(t)=\mathbf{Z}\right\},
\] where the expectation is over the channel distribution and possible randomizations of the control policy. Using the queue evolution equations \eqref{eq:userQueues}, \eqref{eq:codewordQueues}, \eqref{eq:utilityQueues}  and the fact that $([x]^+)^2\leq x^2$, we have
\begin{align}\nonumber 
\Delta L(\mathbf{Z}(t)) \leq B & + \sum_{\Ic \in 2^{\Kc}}\frac{Q_{\Ic}(t)}{F^2}\mathbb{E}\left\{\sum_{\Jc: \Ic\subseteq \Jc}b_{\Ic,\Jc}\sigma_{\Jc}(t) - T_{\rm slot}\mu_{\Ic}(t)\bigg| \mathbf{Z}(t)\right\}\\ \label{eq:drift1_infiniteDemand}
& + \sum_{k=1}^KS_k(t)\mathbb{E}\left\{a_k(t) - \sum_{\Ic: k\in\Ic}\sigma_{\Ic}(t)\bigg| \mathbf{Z}(t)\right\} 
+ \sum_{k=1}^KU_k(t)\mathbb{E}\left\{\gamma_k(t) - a_k(t)| \mathbf{Z}(t)\right\}  
,\end{align}
where
\begin{align} \label{eq:Beta}  
B = \sum_{k=1}^K\left(\gamma^2_{k,\max} + \frac{1}{2}\left(\sum_{\Ic: k\in\Ic}\sigma_{\max}\right)^2 \right)
 + \frac{1}{2F^2}\sum_{\mathcal{I}\in 2^{\Kc}}\sum_{\mathcal{J}:\mathcal{I}\subseteq \mathcal{J}}\left(\sigma_{\max}b_{\Ic,\Jc}\right)^2  + \frac{T_{\rm slot}^2}{2F^2}\sum_{\mathcal{I}\in 2^{\Kc}}\sum_{k\in\mathcal{I}}\mathbb{E}\left\{\left(\log_2(1 + Ph_k(t))\right)^2\right\}.
\end{align}
Note that $B$ is a finite constant that depends only on the parameters of the system. Adding the quantity $-V\sum_{k=1}^K\mathbb{E}\left\{g_k(\gamma_k(t))| \mathbf{Z}(t)\right\}$ to both hands of \eqref{eq:drift1_infiniteDemand} and rearranging the right hand side, we have the drift-plus-penalty expression
\begin{align}\nonumber
\Delta L(\mathbf{Z}(t)) - V\sum_{k=1}^K\mathbb{E}\left\{g_k(\gamma_k(t))| \mathbf{Z}(t)\right\} \leq B  &+\sum_{k=1}^K\mathbb{E}\left\{-Vg_k(\gamma_k(t)) + \gamma_k(t)U_k(t)| \mathbf{Z}(t)\right\}  \\ \nonumber
&+\sum_{\Jc\in 2^{\Kc}}\mathbb{E}\left\{\sigma_{\Jc}(t)| \mathbf{Z}(t)\right\} \left(\sum_{\Ic:\Ic\subseteq \Jc}\frac{Q_{\Ic}(t)}{F^2}b_{\Ic, \Jc}- \sum_{k:k\in\Jc}S_k(t)\right) \\ \nonumber 
& + \sum_{k=1}^K\left(S_k(t) -  U_k(t)\right)\mathbb{E}\left\{a_k(t)| \mathbf{Z}(t)\right\}  \\ \label{eq:drift2_inifiniteDemand}
&-\sum_{\Jc\in 2^{\Kc}}\frac{Q_{\Jc}(t)}{F^2}T_{\rm slot}\mathbb{E}\left\{\mu_{\Jc}(t)| \mathbf{Z}(t)\right\}
\end{align} 

Now observe that the proposed scheme $\pi$ minimizes the right hand side of \eqref{eq:drift2_inifiniteDemand} given any channel state $\hv(t)$ (and hence in expectation over the channel state distributions). Therefore, for every vectors $\overline{\av}\in [1,\gamma_{\max}]^K, \overline{\gammav}\in [1,\gamma_{\max}]^K, \overline{\mathbf{\sigmav}}\in Conv(\{0,..,\sigma_{\max}\}^M),  \overline{\muv}\in \sum_{\hv\in\mathcal{H}}\phi_{\hv}\Gamma(\hv)$ that denote time averages of the control variables achievable by any static (i.e. depending only on the channel state realizations) randomized policies it holds that 
\begin{align}\nonumber
\Delta L^{\pi}(\mathbf{Z}(t)) - V\sum_{k=1}^K\mathbb{E}\left\{g_k(\gamma^{\pi}_k(t))\right\}\leq 
B&- V\sum_{k=1}^Kg_k(\overline{\gamma}_k)  + \sum_{k=1}^KU_k(t)\left(\overline{\gamma}_k - \overline{a}_k\right)+\sum_{k=1}^KS_k(t)\left(\overline{a}_k -\sum_{\Jc :k\in\Jc}\overline{\sigma}_{\Jc}\right)\\ \label{eq:infinite_drift3} 
& + \sum_{\Jc}\frac{Q_{\Jc}(t)}{F^2}\left(\sum_{\Ic :\Jc\subseteq \Ic}b_{\Jc,\Ic}\overline{\sigma}_{\Ic} - T_{\rm slot}\overline{\mu}_{\Jc}\right)
\end{align}
We will use \eqref{eq:infinite_drift3} to compare our policy with the specific static policies defined in Lemmas \ref{lem:StaticOptimalPolicy}, \ref{lem:StablePolicy}.

\textbf{Proof of strong stability:}
Replacing the time averages we get from the static stabilizing policy $\pi^{\delta}$ of Lemma \ref{lem:StablePolicy} for some $\delta>0$, we get that there exist $\epsilon,\epsilon' >0$ such that (the superscript $\pi$ denotes the quantities under our proposed policy)
\[\Delta L^{\pi}(\mathbf{Z}(t)) \leq B + V\sum_{k=1}^K\mathbb{E}\left\{g_k(a^{\pi}_k(t))\right\}- V\sum_{k=1}^Kg_k(\overline{a}_k^{\delta}) -\epsilon \left(\sum_{k=1}^K S_k(t) + \sum_{\Jc\in 2^{\Kc}}\frac{Q_{\Jc}(t)}{F^2}\right)-\epsilon'\sum_{k=1}^KU_k(t).\]
Since $a_k(t)\leq \gamma_{\max,k} \forall t$, it follows that $g_k(\overline{a}_k^{\pi})<g_k(\gamma_{\max,k})$. In addition, $g_k(x)\geq0,\forall x\geq 0$ therefore 
\[
\Delta L^{\pi}(\mathbf{Z}(t)) \leq B + V\sum_{k=1}^K\mathbb{E}\left\{g_k(\gamma_{\max,k})\right\} -\epsilon \left(\sum_{k=1}^K S_k(t) + \sum_{\Jc\in 2^{\Kc}}\frac{Q_{\Jc}(t)}{F^2}\right) -\epsilon'\sum_{k=1}^KU_k(t)
\]
Using the the Foster-Lyapunov criterion, the above inequality implies that the system $\mathbf{Z}(t)(\mathbf{S}(t), \mathbf{Q}(t), \mathbf{U}(t))$ under our proposed policy $\pi$ has a unique stationary probability distribution, under which the mean queue lengths are finite \footnote{For the utility-related virtual queues, note that if $g_k'(0)<\infty$, then $U_k(t)<Vg'_k(0)+\gamma_{k,\max}$, i.e. their length is deterministically bounded.}. Moreover, 
\begin{equation}\label{eq:delayProofEq}  
\lim\sup\limits_{T\rightarrow\infty}\frac{1}{T}\sum_{t=0}^{T-1}\mathbb{E}\left\{\sum_{\Jc\in 2^{\Kc}}\frac{Q_{\Jc}(t)}{F^2}+\sum_{k=1}^K(S_k(t) + U_k(t))\right\} \leq \frac{B+ V\sum_{k=1}^Kg_k(\gamma_{\max,k})}{\epsilon}
.\end{equation}

Therefore the queues are strongly stable under  our proposed policy. In order to prove  the part of Theorem \ref{th:optimality_infinite} regarding the guaranteed bound on the average queue lengths, we first note that the above inequality holds for every $\epsilon > 0$ and define $\epsilon_0$ as 
\begin{align} \label{eq:epsilon1}
\epsilon_0 = & \argmax_{\epsilon > 0}\epsilon \\ \label{eq:epsilon2}
\text{s.t. }& \epsilon \mathbf{1} \in \Lambda^{CC}  
. \end{align}
Following the same arguments as in Section IV of \cite{Li05}, we can show that the Right Hand Side of \eqref{eq:delayProofEq} is bounded from below by 
\[
\frac{B+ V\sum_{k=1}^Kg_k(\gamma_{\max,k})}{\epsilon_0}
,\] 
therefore proving the requested bound on the long-term average queue lengths. 

We now proceed to proving the near-optimality of our proposed policy. 

\textbf{Proof of near optimal utility:} Here we compare $\pi$ with the static optimal policy $\pi^*$ from Lemma \ref{lem:StaticOptimalPolicy}. Since $\pi^*$ takes decisions irrespectively of the queue lengths, we can replace quantities $\overline{\av}, \overline{{\sigmav}},  \overline{\muv}$ on \eqref{eq:infinite_drift3} with the time averages corresponding to $\pi^*$, i.e. $\overline{\av}^*, \overline{\sigmav}^*,  \overline{\muv}^*$. From the inequalities in Lemma \ref{lem:StaticOptimalPolicy} we have  
\[
V\sum_{k=1}^K\mathbb{E}\left\{g_k(\gamma^{\pi}_k(t))\right\} \geq V \sum_{k=1}^Kg_k(\overline{a}^{*}_k) - B + \Delta L^{\pi} (\mathbf{Z}(t))
\]
Taking expectations over $\mathbf{Z}(t)$ for both sides and summing the inequalities for $t=0,1,..,T-1$ and dividing by $VT$ we get 
\[\frac{1}{T}\sum_{t=1}^{T-1}\sum_{k=1}^K\mathbb{E}\left\{g_k(\gamma^{\pi}_k(t))\right\} \geq \sum_{k=1}^Kg_k(\overline{a}^{*}_k) - \frac{B}{V} - \frac{\mathbb{E}\left\{L^{\pi}(\mathbf{Z}(0))\right\}}{VT} + \frac{\mathbb{E}\left\{L^{\pi}(\mathbf{Z}(T))\right\}}{VT}.\]
Assuming $\mathbb{E}\left\{L^{\pi}(\mathbf{Z}(0))\right\} < \infty$ (this assumption is standard in this line of work, for example it holds if the system starts empty), since $\mathbb{E}\{L^{\pi}(\mathbf{Z}(T))\}>0,\forall T>0$, taking the limit as $T$ goes to infinity gives
\[
\lim\limits_{T\rightarrow\infty}\frac{1}{T}\sum_{t=1}^{T-1}\sum_{k=1}^K\mathbb{E}\left\{g_k(\gamma^{\pi}_k(t))\right\} \geq \sum_{k=1}^Kg_k(\overline{a}^{*}_k) - \frac{B}{V}
\]
In addition, since $g_k(x)$ are concave, Jensen's inequality implies 
\[
\sum_{k=1}^Kg_k(\overline{\gamma}_k^{\pi}) =\sum_{k=1}^Kg_k\left(\lim_{T\rightarrow \infty}\frac{1}{T}\sum_{t=0}^T\mathbb{E}\{\gamma_k^{\pi}(t)\}\right)
\geq \lim\limits_{T\rightarrow\infty}\frac{1}{T}\sum_{t=1}^{T-1}\sum_{k=1}^K\mathbb{E}\left\{g_k(\gamma^{\pi}_k(t))\right\}\geq \sum_{k=1}^Kg_k(\overline{a}^{*}_k) - \frac{B}{V}
.\]
Finally, since the  virtual queues $U_k(t)$ are strongly stable, it holds $\overline{a}_k^{\pi} > \overline{\gamma}_k^{\pi}$. We then have 
\[
\sum_{k=1}^Kg_k(\overline{a}_k^{\pi}) > \sum_{k=1}^Kg_k(\overline{\gamma}_k^{\pi}) \geq \sum_{k=1}^Kg_k(\overline{a}^{*}_k) - \frac{B}{V}
,\]
which proves the near optimality of our proposed policy $\pi$. 

\bibliographystyle{IEEEtran}
\bibliography{caching}

\end{document}